\newif\iffull
\fulltrue
\iffull
\documentclass[11pt]{llncs}
\else
\documentclass[runningheads]{llncs}
\fi

\iffull
\usepackage[margin=1in]{geometry}
\usepackage{hyperref}
\fi

\usepackage[english]{babel}
\usepackage[utf8x]{inputenc}
\usepackage[T1]{fontenc}
\usepackage{ifxetex}

\usepackage{amsthm,amssymb}
\usepackage{amsmath}
\usepackage{subcaption} 
\usepackage{graphicx}
\usepackage{pstricks,pst-node}
\usepackage[colorinlistoftodos,disable]{todonotes}
\usepackage{microtype}
\usepackage{cite}
\usepackage{enumerate}
\usepackage{xspace}

\spnewtheorem{claim}{Claim}{\bfseries}{\rmfamily}

\renewcommand{\leadsto}{\rightarrow}  
\newcommand{\gridminor}{grid-major\xspace}
\newcommand{\Gridminor}{Grid-major\xspace}
\newcommand{\SLh}{\ensuremath{\textit{SLh}}}
\newcommand{\PLh}{\ensuremath{\textit{PLh}}}
\newcommand{\VRh}{\ensuremath{\textit{VRh}}}
\newcommand{\HH}{\ensuremath{\textit{Hh}}}
\newcommand{\sHH}{\ensuremath{\textit{sHh}}}
\newcommand{\gmh}{\ensuremath{\textit{GMh}}}
\newcommand{\sgmh}{\ensuremath{\textit{sGMh}}}
\newcommand{\pw}{\ensuremath{\textit{pw}}}
\newcommand{\calR}{\cal R}
\newcommand{\half}{\tfrac{1}{2}}

\iffull
\else
\fi

\title{Homotopy height, \gridminor height and graph-drawing height\thanks{Erin Chambers was supported in part by NSF grants CCF-1614562 and DBI-1759807. David Eppstein was supported in part by NSF grants  CCF-1618301 and CCF-1616248. Arnaud de Mesmay was supported in part by grants ANR-18-CE40-0004-01 (FOCAL) and  ANR-16-CE40-0009-01 (GATO).
This work began at the Fifth Annual Workshop on Geometry and
Graphs, at the Bellairs Research Institute of McGill University.}}
\author{Therese Biedl\inst{1} \and Erin Wolf Chambers\inst{2} \and David Eppstein\inst{3} \and Arnaud De Mesmay\inst{4} \and Tim Ophelders\inst{5}}
\authorrunning{Biedl, Chambers, Eppstein, de Mesmay, Ophelders}

\institute{David R. Cheriton School of Computer Science, University of Waterloo\and
Department of Computer Science, Saint Louis University\and
Computer Science Department, University of California, Irvine\and
Univ. Grenoble Alpes, CNRS, Grenoble INP, GIPSA-lab, 38000 Grenoble, France\and
Department of Computational Mathematics, Science and Engineering, Michigan State University}

\begin{document}
\maketitle

\begin{abstract}
It is well-known that both the pathwidth and the outer-planarity
of a graph can be used to obtain lower bounds on the height of a 
planar straight-line drawing of a graph.  But both bounds fall short
for some graphs.  In this paper, we consider two other parameters,
the (simple) homotopy height and the (simple) \gridminor height.
We discuss the relationship between them and to the other parameters,
and argue that they give lower bounds on the straight-line drawing
height that are never worse than the ones obtained from pathwidth
and outer-planarity.
\end{abstract}

\section{Introduction}

Straight-line drawings of planar graphs are one of the oldest and
most intensely studied problems in graph drawing
\cite{Wag36,Fary48,Stein51,Tut63,FPP90,Sch90}.    
It has been known since the 1990s that every planar graph has a 
straight-line drawing of height $n{-}1$ \cite{FPP90,Sch90} and that
some planar graphs require height $\frac{2}{3}n$ if the outer-face
must be respected \cite{DLT84,FPP88}.
Nevertheless many problems surrounding the height of
planar straight-line drawings remain open; it is not even known
whether minimizing height is NP-hard (although the problem is
NP-hard when edges may only connect adjacent rows \cite{HR92}
and it is fixed-parameter tractable in the output height \cite{DFK+08}).

One of the chief obstacles is that very few tools are known for
arguing that a planar graph requires a certain height in all planar
straight-line drawings.  Two graph parameters are  commonly used for this: the pathwidth (as
the height is at least $\pw(G)$ \cite{FLW03,DFK+08}) 
and the outer-planarity (as the height is at least
twice the outer-planarity minus 1 \cite{DLT84,FPP88}); for detailed definitions see
Section~\ref{sec:definitions}.  However, both parameters may be constant in graphs that require linear height~\cite{Bie11}%
\iffull (see also Fig.~\ref{fig:contractHeight}(b))\fi.

In this paper, we study two other graph parameters, the \emph{homotopy
  height} $\HH(G)$ and the \emph{\gridminor height} $\gmh(G)$ and their simple variants $\sHH(G)$ and $\sgmh(G)$.  Roughly speaking, the homotopy height is defined as the minimum $k$ such that a sequence of paths of length at most $k$ sweep the graph\footnote{We note that there are \emph{many} possible variants of homotopy height, all quantifying in slightly different ways the optimal way to sweep a planar graph with a curve. 
  We have chosen here one particular variant that seems to be most suitable for graph drawing purposes, and we only study it for triangulated graphs. We refer the reader to other recent works on this parameter~\cite{cl-othoah-09,hnss-htwyd-16,cmo-coh-18} for further discussion.\iffull of other variants and their complexity.\fi}, while the \gridminor height is the minimum height of a grid of which the
graph is a minor. 
Fig.~\ref{fig:examples} illustrates this and graph parameters used in the paper.
Our simple variants add simplicity constraints to the paths involved in the sweeping or the columns of the \gridminor representation. We show that despite their apparent differences, homotopy height and \gridminor height are equal, and that both the normal and the simple variants are lower bounds on the
graph drawing height.  More precisely, any planar
triangulated graph $G$ has
\begin{equation}
\label{equ:result}
pw (G)
\stackrel{(*)}{\leq}
\HH(G)
= 
\gmh (G)
\stackrel{(*)}{\leq}
\sHH(G)
=
\sgmh (G)
\stackrel{(*)}{\leq} 
\VRh(G)
= \SLh (G),
\end{equation}
where $\VRh(G)$ and $\SLh(G)$ are the minimum height of a 
visibility representation and straight-line drawing
of $G$.  As we will show, the 
inequalities marked with $(*)$ are strict for some planar graphs.
More strongly, the parameters separated by these inequalities
can differ by non-constant factors from each other.

In particular, the (simple) \gridminor height and homotopy-height can both
serve as lower bounds on the height of a straight-line drawing.
For some graphs (e.g.~the one in Fig.~\ref{fig:contractHeight}(b))
this gives a better lower bound than can be achieved via
pathwidth, though not a better lower bound than what was known  \cite{Bie11}.
We should mention that the outer-planarity $op$ is also related to these parameters via 
\begin{equation}
\label{equ:result2}
2op(G)-1
\stackrel{(*)}{\leq} 
\gmh(G),
\end{equation}
so the homotopy-height and \gridminor height can
also can replace outer-planarity as 
lower-bound tool for graph-drawing height, and in fact, provide a convenient
vehicle for unifying both tools.
While these results have not yet led us to new lower bound
results for straight-line drawings, they provide new tools which had not been 
considered previously, and suggest a promising new line of inquiry.

\begin{figure}[ht]
\begin{subfigure}[b]{0.195\linewidth}
\includegraphics[width=\linewidth,page=1]{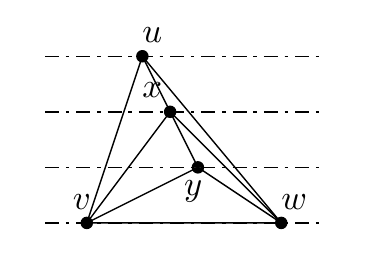}
\caption{}
\end{subfigure}
\begin{subfigure}[b]{0.195\linewidth}
\includegraphics[width=\linewidth,page=2]{examples.pdf}
\caption{}
\end{subfigure}
\begin{subfigure}[b]{0.195\linewidth}
\includegraphics[width=\linewidth,page=3]{examples.pdf}
\caption{}
\end{subfigure}
\begin{subfigure}[b]{0.195\linewidth}
\includegraphics[width=\linewidth,page=4]{examples.pdf}
\caption{}
\end{subfigure}
\begin{subfigure}[b]{0.195\linewidth}
\includegraphics[width=\linewidth,page=5]{examples.pdf}
\caption{}
\end{subfigure}
\caption{The same graph with (a) a straight-line drawing, (b) a flat visibility representation, (c) a simple \gridminor representation, (d) the
corresponding contact-representation, (e) a simple homotopy.  The height is always four.
The moves of the homotopy are: a face-flip at $\{u,v,x\}$, a boundary-move,
an edge-slide at $(v,y)$, a face-flip at $\{x,y,w\}$ and face-flip at $\{u,x,w\}$.
}
\label{fig:examples}
\end{figure}

Our results naturally raise the question of the complexity of computing these parameters. Computing the optimal height of homotopies is conjectured but not known to be NP-hard~\cite{cl-othoah-09}; 
even arguing that it is in NP is non-trivial~\cite{cmo-coh-18}, although it has a logarithmic approximation~\cite{hnss-htwyd-16,cmo-coh-18}.  Our equalities imply that computing the homotopy-height $k=\HH(G)$
is (non-uniform) fixed-parameter tractable in $k$.  Indeed, it equals \gridminor height, which is closed under taking minors. Minor testing can be expressed
in second-order logic, and the graphs of bounded \gridminor height have bounded pathwidth,
so it follows from graph minor theory and Courcelle's theorem~\cite{Courcelle} that for any $k$, the graphs with \gridminor height $k$
\iffull and homotopy-height $k$\fi{} can be recognized in linear time. However, this method
uses the (unknown!) forbidden minors for \gridminor height; finding them remains
an open problem of independent interest. We can also show more directly using Courcelle's theorem that simple \gridminor height is fixed-parameter tractable.

All our results are only for \emph{triangulated planar graphs},
planar graphs where all faces (including the outer-face) are triangles.
This is not a big restriction for graph drawing height, 
as any planar graph $G$ is a subgraph of a triangulated planar graph $G'$ that
has a straight-line drawing of the same height, up to a small additive term.  
(Obtain $G'$ by triangulating the convex hull of a drawing of $G$ and adding
three vertices that surround the drawing.)  Most of our parameters
naturally carry over to non-triangulated planar graphs, but some
parameters would be much more cumbersome to define and work with for non-triangle faces. 

\section{Definitions}
\label{sec:definitions}

All graphs in this paper are \emph{planar}: they can be drawn in the plane
without crossings.  Their \emph{faces} are maximal connected regions
that remain when removing the drawing); we call the unbounded face the
\emph{outer-face}.  
%
Unless otherwise stated, we study only \emph{simple} graphs that have no loops and at most one edge between any two vertices, and we almost
always study \emph{triangulated graphs}, where all faces
(including the outer-face) are bounded by a simple cycle of length 3.
Such a graph is \emph{maximal planar}: no edge can be added without violating
simplicity or planarity. Its planar embedding is unique up to the choice of outer-face.

Let $G$ be a triangulated graph with fixed outer-face $f$.
We define \emph{outer-planarity} $op(G)$ via a removal process
as follows:  In a first step, remove all vertices on the outer-face.
In each subsequent step, remove all vertices on the outer-face
of the remaining graph.  Then $op(G,f)$ is the number of steps until
no vertices remain, and $op(G)$ is the minimum of $op(G,f)$ over all
choices of face $f$.

\paragraph*{Graph-drawing parameters:} 

The \emph{$W\times H$-grid} has
vertices at the \emph{grid-points} $\{1,\dots,W\} \times \{1,\dots,H\}$  and
an edge between any two grid-points of distance one.
A \emph{straight-line drawing} of $G$ consists
of a mapping of $G$ to grid-points such that if all edges are drawn as
straight-line segments between their endpoints, no two edges cross and no
edge overlaps a non-incident vertex.  Every planar
graph has such a drawing \cite{Wag36,Fary48,Stein51} whose supporting grid has height at most $n-1$ \cite{FPP90,Sch90}.
We use $\SLh(G)$ to denote the smallest height $h$ of 
a straight-line drawing of $G$.

A \emph{flat visibility-representation} of $G$ consists of an assignment of
a horizontal segment (\emph{bar}) to every vertex  of $G$ such that for any 
edge $(v,w)$ there exists a \emph{line of visibility}, i.e., a line segment
connecting bars of $v,w$ that intersects no other bar.
In the original definition lines of visibility had to be horizontal;
for us it will be more convenient to allow both horizontal and vertical
lines of visibility, as long as they do not cross.
Every planar graph has a flat visibility representation
\cite{Wis85,TT86,RT86}.
We use $\VRh(G)$ to denote the smallest height $h$ of such a representation,
presuming all bars reside at positive integral $y$-coordinates.

\paragraph*{Width parameters:}
%
%
A \emph{path decomposition} of a graph $G$ is a collection
$X_1,\dots,X_L$ of vertex-sets (\emph{bags}) that satisfies the following:
each vertex $v$ appears in at least one bag, the bags containing $v$ are
consecutive, and for each edge $(v,w)$ at least one bag contains both $v$
and $w$.  The \emph{width} of a path decomposition of $G$ is the largest bag-size
minus 1, and the \emph{pathwidth} $\pw(G)$ of a graph is the smallest possible
width of a path decomposition.

We introduce another width parameter which is quite natural, but 
to our knowledge has not been studied before.  
A \emph{grid-representation}
of a graph $G$ consists of a $W\times H$-grid where each gridpoint is labelled
with one vertex of $G$ in such a way that (1) every vertex appears at least
once as a label, (2) for any vertex $v$ the grid-points that are labelled $v$
induce a connected subgraph of the grid, and (3) for any edge $(v,w)$ of $G$
there exists a grid-edge where the ends are labelled $v$ and $w$.  In
particular, if $G$ has a grid-representation then it is a minor of the
$W\times H$-grid.  
Let $\gmh(G)$ be the 
\emph{\gridminor height}, i.e., the smallest $h$ such that $G$ has
a grid-representation where the grid has height $h$.

We say that a \gridminor representation of $G$ is \emph{simple} if in 
every column $c$ of the grid and for any vertex $v$ of $G$, the nodes 
labeled $v$ in $c$ form a \iffull connected subgraph (which is thus a path)\fi.  
The \emph{simple 
\gridminor height} of $G$, denoted $\sgmh(G)$, is the smallest $h$
such that $G$ has a simple \gridminor representation of height $h$.

A \gridminor representation of height $h$ can be viewed, equivalently, as a \emph{contact-rep\-re\-sen\-ta\-tion}
with integral orthogonal polygons as follows:  Assign to every vertex $v$ 
the polygon $P(v)$ that we obtain if we replace every grid-point
labelled $v$ with a unit square centered at that grid-point and take their
union.  
Since the
grid-representation uses integral points, the coordinates of sides of $P(v)$ are halfway between integers.
See Fig.~\ref{fig:examples}(d).  
We get a set of interior-disjoint
orthogonal polygons with integer edge-lengths whose union is a rectangle 
of height $h$, where $(v,w)$ is an edge of $G$ if and only if $P(v)$ and $P(w)$
share at least one unit-length segment on their boundaries.  Conversely
any contact-representation with integral orthogonal polygons 
that uses all points inside a bounding rectangle
can be viewed as a \gridminor representation.  A simple \gridminor representation
becomes a contact representation with \emph{$x$-monotone polygons}
(every vertical line intersects the polygon in an interval) and vice versa.
Contact-representations
of graphs have been studied extensively (see e.g.~\cite{ABF+13} and the
references therein), but to our knowledge the question of the required
height of such representations has not previously been considered.

\paragraph*{Homotopy parameters:}
\iffull A homotopy of a triangulated graph $G$ is a sequence of paths, connected by elementary moves, that together sweep the entirety of the graph $G$. We will allow these paths to have moving endpoints on the outer-face of $G$. Precisely, a\fi \emph{(discrete) simple homotopy} is defined for a planar triangulated graph $G$ with a fixed outer-face $\{u,v,w\}$, and it consists of a sequence $h_0,\dots,h_W$ of walks in $G$ (we call these \emph{curves}) such that:
\begin{enumerate}
\item $h_0$ and $h_W$ are trivial curves at two distinct vertices of the outer-face, say $u$ and $v$.
      \item The vertices $u$ and $v$ partition the outer-face into two subpaths $s(uv)$ and $t(uv)$. For $0\leq i\leq W$, the curve $h_i$ starts on $s(uv)$ and ends on $t(uv)$. 
\item For all $0\leq i < W$
  we can obtain $h_{i+1}$ from $h_i$ with a face-flip, edge-slide, a boundary-move or a boundary-edge-slide%
\iffull
; see Fig.~\ref{fig:homotopymoves}.
\fi

\end{enumerate}

\iffull
\begin{figure}\hspace*{\fill} %
\includegraphics[page=2]{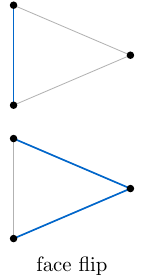}\hspace*{\fill} 
\includegraphics[page=4]{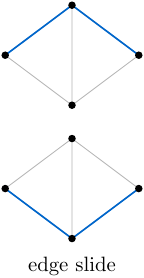}\hspace*{\fill} 
\caption{An illustration of the types of moves in a discrete homotopy.}
\label{fig:homotopymoves}
\end{figure}

\todo[inline]{TB: The second kind of face-flip is weird!  I guess this matches
the definition, but if you disallow the first time of face-flip (resulting
in a spike) in this
situation then the proof of Lemma 5 needs a bit more work.  Specifically,
Fig. 5(c) is a face-flip for simple homotopies, but might be ``half''
of this second kind of face-flip if the vertex below the junction also
shows up higher up again.}

\todo[inline]{A:I think that we should disallow the second kind of face-flip because it is the concatenation of a spike and the first kind of face flip. Actually I think that the definition already disallows it (``that contains edges of $h_i$''), so I removed it.}

\fi

Here a
\emph{face-flip} consists of picking an inner face $\{x,y,z\}$ such that the subsequence $x$-$y$ is in $h_i$, and replacing the sub-path $x$-$y$ by $x$-$z$-$y$ to obtain $h_{i+1}$. The reverse move, going from $x$-$z$-$y$ to $x$-$y$, is also allowed.
An \emph{edge-slide}\footnote{Edge-slides are typically not allowed in discrete homotopies, but the result of one edge-slide is the same as flipping two inner faces consecutively. 
Thus, allowing edge-slides only results in an additive difference of at most one for the homotopy height.} consists
of picking an edge $e=(x,y)$ adjacent to two inner faces $\{x,y,z\}$ and $\{x,y,t\}$, such that the subsequence $z$-$x$-$t$ is in $h_i$. Then 
replace the subpath $z$-$x$-$t$ in $h_i$ by $z$-$y$-$t$ to obtain $h_{i+1}$.  
A \emph{boundary-move} consists of picking an edge $e=(x,y)$ on the outer face, and, if $e \in s(uv)$, and $x$ is the start of $h_i$, it appends $y$ so that it becomes the new starting point (thus replacing $x$ by the subsequence $y$-$x$). If $e \in t(uv)$ and $x$ is the end of $h_i$, it appends $y$ at the end. The reverse operations are also allowed.
A \emph{boundary-edge-slide} consists of picking an edge $e=(x,y)$ on the outer face adjacent to an inner face $\{x,y,z\}$, and, if $e \in s(uv)$ and $h_i$ starts with $x$-$z$, we flip $\{x,y,z\}$ and remove $e$, i.e., we replace the starting subsequence $x$-$z$ by $y$-$x$. The symmetric operation for edges on $t(uv)$ is also allowed. (Observe that this boundary-edge-slide is the same as flipping a face and removing the boundary edge with a boundary move.). 
\iffull
See Fig.~\ref{fig:homotopymoves}.
\fi



The \emph{height} of a simple homotopy is the length of the longest path $h_i$,
counting as path-length the number of vertices.  Let $\sHH(G,f)$ be the
minimum height of a simple homotopy of $G$ that uses $f$ as outer-face, and
set $\sHH(G)$ (the \emph{simple homotopy height}) to be the minimum of $\sHH(G,f)$ 
over all choices of outer-faces $f$.  (Since we only study triangulated
graphs the rotation scheme is unique and so this covers all possible
planar embeddings.)  

The definition of a (non-simple) homotopy is obtained by removing the simplicity assumption on the curves $h_i$, and allowing two other kinds of moves (spikes and unspikes) leveraging the non-simplicity of the curves. For technical reasons and to obtain a maximal generality, we will also relax the conditions on the endpoints and the starting and ending curves. Since the precise definition is somewhat technical, we postpone it to Section~\ref{sec:hhgmh}.

\paragraph*{Some simple results:}
We briefly review some relationships that are well-known, or easily derived.

\begin{itemize}
\item $\pw(G)\leq \gmh(G)$ since a $W \times H$-grid has pathwidth at most $H$ 
	and pathwidth is closed under taking minors.
\item Obviously $\gmh(G)\leq  \sgmh(G)$.
\item $\sgmh(G)\leq \VRh(G)$ since a flat visibility
	representation can easily be converted into a simple \gridminor
	representation by assigning label $v$ to all grid-points 
	of the bar of $v$ as well as all grid-points that this bar can
	see downward or rightward without intersecting other bars or
	non-incident edges.   See Fig.~\ref{fig:examples}(b-c).
\item $\VRh(G)=\SLh(G)$  since flat visibility
	representations can be transformed into straight-line drawings
	of the same height, and vice versa (\cite{Bie14}.
	using \cite{PT04,EFLN06}).
\item Finally we have $2op(G)-1\leq \gmh(G)$.  To this end, assume that
	we have a \gridminor representation $\Gamma$ of $G$ of height $h$.
	Observe that the grid-graph $\Gamma$ has outer-planarity 
	$\lceil  h/2 \rceil$.  Since outer-planarity does not increase when
	taking minors it follows that $op(G)\leq \lceil h/2 \rceil \leq \frac{h+1}{2}$.
\end{itemize}

\section{Homotopy-height and \gridminor height}

The above inequalities fill in most of the chain in 
Equation~\ref{equ:result}, but one key new part is missing: how does
the (simple) homotopy height relate to the (simple) \gridminor height?

\subsection{Simple \gridminor height and simple homotopy height}

\begin{lemma}
\label{lem:HH_gm}
\label{lem:HHsm_gms}
For any triangulated planar graph $G$ we have
$\sgmh(G)\leq \sHH(G)$.
\end{lemma}
\begin{proof}
Let $h_0,\dots,h_W$ be a simple homotopy of height $k=\sHH(G)$.
The rough idea is to label a $W\times k$-grid by giving the gridpoints in
the $i$th column the labels of the vertices in $h_i$, in top-to-bottom-order,
and adding some duplicate copies of vertices in $h_i$ to fill the column.
However, we must insert more columns in between to ensure the properties of a
simple \gridminor representation.

It will be easier to describe this by giving a contact-representation
of $G$ where all polygons are $x$-monotone and the height is $k$.
Curve $h_0$ is a vertex $u$; we initialize $P(u)$ as a $1 \times k$ rectangle.
Now assume
that for some $i\geq 0$, we have built a contact representation $\Gamma_i$
of the graph that was swept by $h_0,\dots,h_i$.  Furthermore, the
right boundary of the bounding box of $\Gamma_i$ contains sides of
the vertices in $h_i$, in order.  
\iffull
We explain how to expand $\Gamma_i$ 
rightwards to represent $h_{i+1}$, depending on the move that was used
to obtain $h_{i+1}$, see also Fig.~\ref{fig:moves}.

\begin{itemize}
\item Assume first that $h_{i+1}$ was obtained from an edge-slide along $(x,y)$,
	thus $h_i$ contained $t$-$x$-$z$ for some vertices $t,z$ and this
	was replaced by $t$-$y$-$z$ in $h_{i+1}$.  Expand $\Gamma_i$ 
	rightward by one unit, and expand all polygons of vertices in $h_i$
	rightward by one unit, except for polygon $P(x)$.  The pixels 
	to the right of $P(x)$ are used to start a new polygon $P(y)$.
\item A boundary-edge slide is handled in exactly the same way; the only
	difference is that $x$ and $y$ are vertices on the outerface and vertex $z$ does not exist.
\item Assume now that $h_{i+1}$ was obtained from a face-flip at
face $\{a,b,c\}$.  Assume further that it was \emph{increasing}, i.e.,
$|h_{i+1}|>|h_i|$.  Up to renaming,
path $a$-$b$ in $h_i$ was replaced
by path $a$-$c$-$b$ in $h_{i+1}$.
So $|h_{i}|=|h_{i+1}|{-}1\leq k{-}1$, which means that at least one
polygon $P(d)$ for some $d\in h_i$ has height 2 or more on the right
boundary.  Insert an (upward or downward) staircase that shifts this
extra height to $a$ or $b$ while keeping all polygons of $h_i$ on the
right boundary, see Fig.~\ref{fig:moves}.  This may take
up to $k$ units of width, but does not affect the height.

Now one of $P(a),P(b)$ has height 2 or more on the right
boundary.  Expand the drawing
	rightward by one more unit, and expand  all polygons
	rightward by one unit, except that at the horizontal
	line between $P(a)$ and $P(b)$ we remove one pixel
	from the polygon that has height at least 2 and start
	polygon $P(c)$ there.

\item Assume now that $h_{i+1}$ was obtained from a face-flip at
face $\{a',b',c'\}$  that was \emph{decreasing}, i.e.,
$|h_{i+1}|<|h_i|$.  Up to renaming, $h_i$ contains $a'$-$c'$-$b'$
which was replaced by $a'$-$b'$ in $h_{i+1}$.  Expand $\Gamma_i$
	rightward by one unit, and expand all polygons of vertices in $h_i$
	rightward by one unit, except for polygon $P(c')$.  The pixels 
	to the right of $P(c')$ added to $P(a)$ or $P(b)$.
\item An (increasing or decreasing) boundary-move is handled similarly; the only
	difference is that $c=v$ or $c=w$ and $b$ does not exist.
\end{itemize}
\fi
\todo{TB: time-permitting, I want to add an illustration what to do for a spike and maybe also show the corresponding graph with the curves in it.}

\begin{figure}[ht]
\hspace*{\fill}
\includegraphics[width=0.45\linewidth,page=2]{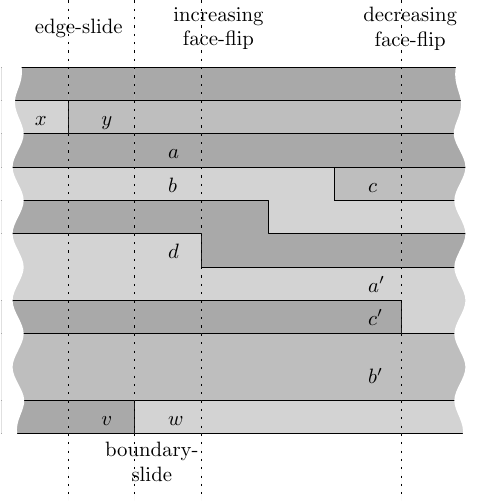}
\hspace*{\fill}
\caption{Converting a discrete simple homotopy into a contact-representation.}
\label{fig:HH_to_gm}
\label{fig:step_down}
\label{fig:moves}
\end{figure}

In all cases the polygons remain connected and are $x$-monotone.
Furthermore we realized exactly those incidences that were added to the
graph when sweeping to $h_{i+1}$, and the right boundary contains exactly
the polygons of vertices of $h_{i+1}$, in order.
Therefore, repeating gives a contact-representation of height $k$ that uses $x$-monotone polygons, and thus the desired simple \gridminor representation.
\end{proof}

\begin{lemma}
\label{lem:gms_HHsm}
For any triangulated planar graph $G$ we have
$\sHH(G)\leq \sgmh(G)$.
\end{lemma}
\begin{proof}
Fix a simple grid representation $\Gamma$ of $G$ of height $\sgmh(G)$.
For this proof it will be easier to interpret $\Gamma$ as a contact representation.
So for any vertex $z$, let $P(z)$ be the orthogonal polygon obtained by
taking the unions of all unit squares whose centerpoint is a grid point
labelled $z$.  
Since the
grid-representation uses integral points, the coordinates of sides of $P(v)$ are halfway between integers.
\iffull
$P(z)$ is $x$-monotone
since $\Gamma$ is simple, and it therefore has a unique
\emph{leftmost side}, i.e., vertical side of $P(z)$ that minimizes the
$x$-coordinate.

The idea of obtaining a simple homotopy from this is easy:  For each integer
$i$, let $h_i$ be the set of vertices whose polygons intersect the vertical 
line $\{x=i\}$, listed from top to bottom.  
This is clearly a walk in the contact-graph of $\Gamma$, and
hence a walk in $G$ since $G$ is maximal.  It is also simple since 
polygons are $x$-monotone.  However, to argue that these curves satisfy
the assumptions on a homotopy, we first need to modify the contact 
representation (without changing its heights) to satisfy additional
properties.
\fi

A \emph{junction} is a point that belongs to
at least three sides of polygons; 
we call it \emph{interior/exterior} depending on
whether it lies on the boundary of 
the rectangle $\calR$ that encloses the contact representation.
No junction can belong to four
sides since $G$ is maximal,
so it can be classified as \emph{horizontal} or \emph{vertical}
depending on the majority among its incident sides. 
A \emph{corner} is a point that belongs to exactly two sides of polygons%
\iffull 
; thus this is either a corner of $\calR$ or a place where a
reflex corner of one polygon $P(v)$ equals a convex corner of an adjacent
polygon $P(v')$.     
A \emph{side} of the contact representation is 
a horizontal or vertical line segments that connects two junctions or corners.
(This may be equal to or a subset of a side of polygon, but the meaning of
``side'' will be clear from context.)
\fi
It is not possible for both ends of a side to be exterior junctions, 
or else the corresponding
edge of $G$ would be a bridge of the graph, contradicting the 3-connectivity of
the triangulated graph $G$.

\iffull

\begin{claim}
\label{cl:vertical_junction}
We may assume that $\Gamma$ has no interior vertical junction.
\end{claim}
\begin{proof}
An interior vertical junction can be replaced by a corner and an interior
horizontal junction  (after lengthening some other horizontal edges)
without changing
the height or the  adjacencies.  
See Fig.~\ref{fig:verticalJunction}.
This change can be done in two ways. Let $P(z)$ be the polygon
whose side $s$ includes both vertical sides incident to the 
junction.  If $P(z)$ is $x$-monotone then $s$ either on the top or
the bottom chain (or it is leftmost or rightmost, but then both
methods of doing the change work).    If $s$ was part of the top
chain of $P(z)$, then the change where $P(z)$ is below the new
horizontal edge maintains $x$-monotonicity, else the other one does.
All other polygons have no new horizontal sides and so clearly stay $x$-monotone.
\end{proof}

\begin{claim}
\label{cl:aligned_sides}
We may assume that no two vertical sides have the same $x$-coordinate unless
they are both on the left boundary or both on the right boundary.
\end{claim}
\begin{proof}
Let $e$ and $e'$ be two vertical sides with the same
$x$-coordinate, and let $\ell$ be the line through them, where $\ell$ is
not the left or right boundary.  Walking from
$e$ towards $e'$ along $\ell$, we must encounter an end of $e$,
which is a junction or a corner.  If it is a junction then it must be
horizontal since $e$ is not on the left or right boundary.
Either way, there is a region between $e$ and $e'$ on $\ell$
that does not belong to a side
of a polygon.  Therefore we can move $e$ slightly rightward so that
the sides have different $x$-coordinates.  See Fig.~\ref{fig:moveVertical}.
\end{proof}

\begin{figure}[ht]
\hspace*{\fill}
\begin{subfigure}[b]{0.35\linewidth}
\includegraphics[scale=0.5,page=2]{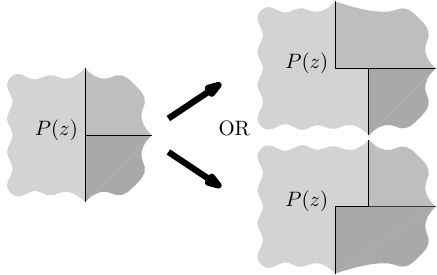}
\caption{}
\label{fig:verticalJunction}
\end{subfigure}
\hspace*{\fill}
\begin{subfigure}[b]{0.3\linewidth}
\includegraphics[scale=0.5,page=3]{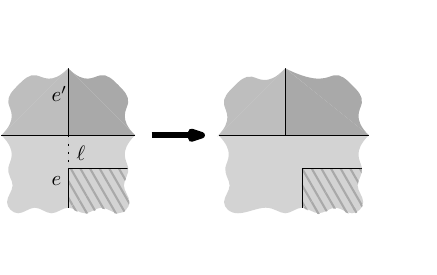}
\caption{}
\label{fig:moveVertical}
\end{subfigure}
\hspace*{\fill}
\begin{subfigure}[b]{0.25\linewidth}
\includegraphics[scale=0.5,page=4,trim=0 0 70 0,clip]{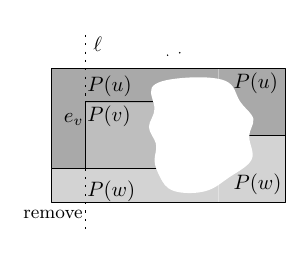}
\caption{}
\label{fig:thirdLeft}
\end{subfigure}
\hspace*{\fill}
\caption{(a) Bringing $P(v)$ to the left boundary. (b) Removing
vertical interior junctions. (c) Making $x$-coordinates distinct.
(d) Bringing a third vertex to the left boundary.
(e) Making the left and right boundary support two distinct edges.
}
\end{figure}

We say that a polygon $P(z)$ 
\emph{touches the top/left/bottom/right boundary} if it contains
part of this boundary of $\calR$.

\begin{claim} 
\label{cl:two_left}
We may assume that at least two polygons touch the left boundary.
\todo{TB: This claim is really helpful below and so got re-inserted.  I 
rewrote the proof to be more similar to the one of the next claim.}
\end{claim}
\begin{proof}
Assume that only one polygon $P(u)$ touches the left boundary.  Let $v\neq u$
be such that the left side $e_v$ of $P(v)$ has minimal $x$-coordinate.
(Fig.~\ref{fig:thirdLeft} illustrates the scenario if one ignores $P(w)$.)
Since $v\neq u$, side $e_v$ is not on the left boundary, hence interior
(except perhaps at its ends).  No junction lies in the open segment $e_v$, 
because this would be an interior vertical
junction and we removed those.
Not both ends of $e_v$ can be corners, else both those corners would 
be convex for $P(v)$ (by choice of $e_v$ as leftmost side), hence reflex for
the polygon $P$ to the left of $e_v$, which contradicts $x$-monotonicity. 
Neither end of $e_v$ can be an interior junction, else the
polygon $P(w)$ with $w\neq u,v$ at it would have points farther left,
contradicting the choice of $v$.  So at least one end of $e_v$ is an exterior
junction and $v$ touches the boundary.  The other end of $e_v$ cannot be an
exterior junction (otherwise $v$ would be a cutvertex) and so it must be a corner,
creating a horizontal side common to $P(v)$ and $P(u)$.  So 
we can remove everything to the 
left of the line $\ell$ through $e_v$ (which contains only points of $P(u)$)
and retain a contact representation of the same graph since $P(u)$ and $P(v)$
have a horizontal side in common.
\end{proof}

\begin{claim}
\label{cl:three_outside}
We may assume that the union of the four boundaries consists of exactly three vertices, while preserving the previous assumption.
\end{claim}
\begin{proof}
 First observe that there cannot be more than three different vertices occupying the left, top right, and bottom boundaries of $\calR$, since $G$ is triangulated.   

\todo{TB: here we're using $x$-monotonicity}
By Claim~\ref{cl:two_left} we can assume that at least two different vertices touch the left boundary, and at least two different vertices touch the right boundary by a symmetric claim.  This proves the claim unless the same two vertices $u,v$ touch both left and right boundary,
which in particular means that they meet at two exterior junctions.
Let $v\neq u,w$ be a vertex that minimizes the $x$-coordinate of its leftmost
side $e_v$.  Exactly as in Claim~\ref{cl:two_left} one argues that not both ends of $e_v$ can be corners,  so one end of
$e_v$ is a junction.  This junction is internal (else $v$ would touch the
top or bottom boundary), so it is horizontal.  The other two polygons at
this junction must be $P(u)$ and $P(w)$ because they occupy points farther left than $v$,
see Fig.~\ref{fig:thirdLeft}. This
implies that the other end of $e_v$ \emph{must} be a corner, else three
vertices would have points to the left of $e_v$, 
contradicting the choice of $v$.  Therefore $P(v)$ is adjacent to both
$P(u)$ and  $P(w)$ at horizontal sides.  We can now delete everything to
the left of the line $\ell$ through $e_v$; this changes no adjacency since 
$v$ is adjacent to $u,w$ via horizontal sides and $(u,w)$ is realized
at the other external junction.
\end{proof}

\begin{claim}
\label{cl:singletons}
We may assume that exactly one vertex touches the left boundary and exactly one different vertex touches the right boundary, while preserving the previous assumptions (except those of Claim~\ref{cl:two_left}).
\todo{TB (new!) We need that those two vertices are \emph{different}.  This was previously said to be obvious, but I really don't see why it is obvious.  Proving it turned into a nasty case-analysis, which is why I re-inserted the removed claim to shorten the argument. A: Added that it does not preserve Claim 3 (which contradicts it)}
  \end{claim}
\begin{proof}
We already know that three vertices touch some boundary, and that
at least two of them touch the left boundary and (by a symmetric argument)
at least two of them touch the right boundary.  Not all three of them can
touch both the left and the right boundary, else at least one of them would
be a cutvertex.  So assume (up to a horizontal flip) that exactly two vertices
touch the right boundary.

If three vertices touch the left boundary, then let $v$ be the ``middle'' one,
i.e., neither of the left corners of $\calR$ belong to $P(v)$.
Otherwise choose $v$ arbitrarily
among the two vertices touching the left boundary.  Let $w\neq v$ be a
vertex that touches the right boundary.
Insert a new column on the left and assign it entirely to $P(v)$, and a new column on the right and assign it entirely to $P(w)$.    Any vertex that previously touched the boundary continues to do so, because it either occupied a corner of $\calR$ (then it remains on the boundary) or it was vertex $v$ (which touches
the new boundary on the left).


Finally we apply Claim~\ref{cl:vertical_junction} and \ref{cl:aligned_sides}
to make interior junctions horizontal and make vertical sides not on the same 
boundary have different $x$-coordinates.
Neither operation affects incidences with the boundary, and the claim holds.
  \end{proof}

We are now ready to extract the homotopy as described before.  We may
assume that exactly three vertices touch the boundary;  they
form a triangular face $f$ in $G$ which we declare to be the outer-face.
Set $u\neq v$ to be the vertices on the left and right boundary.
\fi
As in the definition of homotopy, let $s(uv)$ and $t(uv)$ be the subpaths of $G$ between $u$ and $v$ on $f$. By definition, they consist of the vertices occupying the top and and bottom boundaries of $\calR$.

Let the contact representation now have $x$-range $[-\half,W+\half]$.
For $i=0,\dots,W$, define $h_i$ to be the vertices whose polygons intersect
the vertical line $\{x=i\}$, enumerated from top to bottom.  Clearly 
$h_0=\langle u\rangle$, $h_W=\langle v\rangle$,  and any $h_i$ begins on $s(uv)$ and ends on $t(uv)$.  It remains to show that for
$0\leq i<W$ going
from $h_i$ to $h_{i+1}$ is one of the permitted moves.  Consider some
vertical side $e$ that has $x$-coordinate $i+\half$ (if there is none
then $h_i=h_{i+1}$ and we are done).  
Note that the change from $h_i$ to $h_{i+1}$ affects \emph{only}
vertices that are incident to $e$ or participate in junctions at the
ends of $e$, because no other vertical side has $x$-coordinate $i+\half$
(by $0\leq i<W$ and assumption) and so there is no difference between
the curves elsewhere.
\iffull
Let $P(x)$ and $P(y)$ be the
polygons whose sides contain $e$, with $P(x)$ left of $e$.

We distinguish cases by what the ends of $e$ are:
\begin{enumerate}[(a)]
\item Both ends of $e$ are corners, and the adjacent horizontal sides
	go in opposite directions.  Then $h_i=h_{i+1}$ and we are done.
	See Fig.~\ref{fig:Case1}.
\item Both ends of $e$ are corners, and the adjacent horizontal sides
	go in the same directions.  Say they both go right and note
	that they are interior, else this would be a junction not a corner.
	Therefore the vertical line to the right of $e$ intersects $P(x)$ 
	in two intervals, contradicting $x$-monotonicity.  (Note that in fact we can get
	from $h_i$ to $h_{i+1}$ with a spike or unspike.)
	See Fig.~\ref{fig:Case2}.
\item One end of $e$ is a corner, the other is an interior junction
	(hence horizontal by assumption).  Let $z$ be vertex other than $x,y$
	at the junction, so we have a face $x,y,z$.  The change from
	$h_i$ to $h_{i+1}$ is to replace $x$-$z$ by $x$-$y$-$z$
	or $y$-$x$-$z$ by $y$-$z$; this is a face-flip. 
	See Fig.~\ref{fig:Case3}.
      \item One end of $e$ is a corner, the other is an exterior junction. By the assumptions, this exterior junction is horizontal, and its two vertices, say $v$ and $w$, belong to $t(uv)$ if the junction is on the top boundary, and to $s(uv)$ if the junction is on the bottom boundary. The corresponding change is to replace $v$ by $v$-$w$ at
	the end of the curve, or vice versa; this is a boundary-move. See Fig.~\ref{fig:Case4}.
\item Both ends of $e$ are interior junctions, which we know to be
	horizontal.  Let $z$ and $t$ be the third vertices at these
	junctions.  The change is to replace $t$-$x$-$z$ by $t$-$y$-$z$;
	this is an edge-slide along $(x,y)$.
	See Fig.~\ref{fig:Case5}.
\item Both ends of $e$ are junctions, one is interior (hence horizontal)
	while the other is exterior.  As before, by assumption,  this exterior junction is horizontal, and its two vertices, say $v$ and $w$, belong to $t(uv)$ if the junction is on the top boundary, and to $s(uv)$ if the junction is on the bottom boundary. The change is to
	replace $t$-$v$ by $t$-$w$; this is a boundary edge-slide.
	See Fig.~\ref{fig:Case6}.
\item Both ends of $e$ are exterior junctions. This corresponds to a bridge in $G$, contradicting the fact that it is triangulated.
\end{enumerate}
\fi
Therefore we only use allowed moves and have found a homotopy.  It is simple since polygons are $x$-monotone. 
The
height equals the maximum number of intersected polygons, which is no
more than the height of the contact representation, hence the height
of the grid representation.
\end{proof}

\begin{figure}[ht]
\hspace*{\fill}
\begin{subfigure}[b]{0.15\linewidth}
\includegraphics[width=\linewidth,page=1]{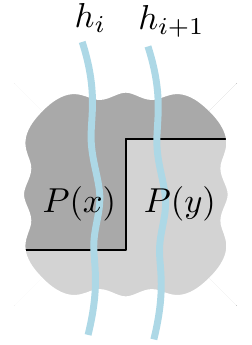}
\caption{}
\label{fig:Case1}
\end{subfigure}
\hspace*{\fill}
\begin{subfigure}[b]{0.15\linewidth}
\includegraphics[width=\linewidth,page=6]{extractHomotopy.pdf}
\caption{}
\label{fig:Case2}
\end{subfigure}
\begin{subfigure}[b]{0.15\linewidth}
\includegraphics[width=\linewidth,page=2]{extractHomotopy.pdf}
\caption{}
\label{fig:Case3}
\end{subfigure}
\begin{subfigure}[b]{0.15\linewidth}
\includegraphics[width=\linewidth,page=3]{extractHomotopy.pdf}
\caption{}
\label{fig:Case4}
\end{subfigure}
\begin{subfigure}[b]{0.15\linewidth}
\includegraphics[width=\linewidth,page=4]{extractHomotopy.pdf}
\caption{}
\label{fig:Case5}
\end{subfigure}
\begin{subfigure}[b]{0.15\linewidth}
\includegraphics[width=\linewidth,page=5]{extractHomotopy.pdf}
\caption{}
\label{fig:Case6}
\end{subfigure}
\caption{The cases of how curves change.}
\label{fig:extractHomotopy}
\end{figure}

Putting the two results together, the homotopy-height and the
simple \gridminor-height are exactly the same value. 

\subsection{Homotopy height and \gridminor height}\label{sec:hhgmh}

One can reasonably argue that the notion of \gridminor height is more natural than simple \gridminor height. Furthermore, it is trivially a minor-closed quantity, which is advantageous from a structural and algorithmic point of view. In this subsection we show that \gridminor height can also be interpreted as the height of a more general notion of homotopy than the one defined in the preliminaries. Compared to the case of simple homotopies, in a (non-simple) homotopy, we remove the hypothesis that the curves are simple and we allow two new moves, spikes and unspikes, leveraging this non-simplicity.   Fig.~\ref{fig:extractHomotopy}(b) illustrates a spike. Furthermore, the conditions are slightly relaxed: the endpoints are allowed to move along an edge instead of a face, and the starting and ending vertices are allowed to be the same. This could allow ``trivial'' homotopies (for example an empty one, idling on a single vertex), and thus we add a new condition on topological non-triviality to disallow those. 

\iffull
The equality with \gridminor height gives a non-trivial algorithmic corollary: since \gridminor height is obviously a minor-closed quantity, it is computable in time FPT in the output, and thus this is also the case for homotopy height. 
The proofs are very similar to those in the previous subsection and we only focus on the salient differences.
\fi

\iffull
Precisely, a \emph{(discrete) homotopy} is defined for a planar triangulated graph $G$ with either a fixed outer-face $\{u,v,w\}$, or a \emph{fixed outer-edge}, which consists of choosing an edge of the graph, doubling it, and fixing the new degree $2$ face between these two edges as the outer face of $G$. It consists of a sequence $h_0,\dots,h_W$
of walks in $G$ (we call these \emph{curves}) such that:
\begin{enumerate}
\item $h_0$ and $h_W$ are trivial curves at two vertices of the outer-face, say $u$ and $v$ ($u=v$ is allowed).
      \item The vertices $u$ and $v$ partition the outer-face into two subpaths $s(uv)$ and $t(uv)$. For $0\leq i\leq W$, the curve $h_i$ starts on $s(uv)$ and ends on $t(uv)$. 
\item For all $0\leq i < W$
  we can obtain $h_{i+1}$ from $h_i$ with a face-flip, spike, unspike, edge-slide, a boundary-move or a boundary-edge-slide%
; see Fig.~\ref{fig:spikes}.
  \item $h$ is \emph{topologically non-trivial} 
(see below).
\end{enumerate}

\fi

\iffull
\begin{figure}\hspace*{\fill} %
\includegraphics[page=1]{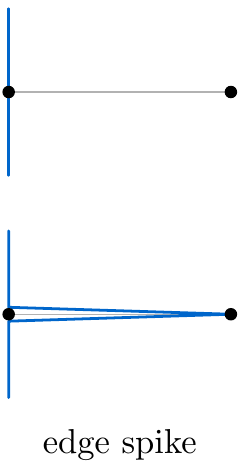}\hspace*{\fill} 
\caption{The last type of homotopy move.}
\label{fig:spikes}
\end{figure}
\fi



\iffull

A \emph{spike} consists of picking some edge $e=(x,y)$ where $x\in h_i$ and replacing $x$ by the subsequence $x$-$y$-$x$
(traversing $e$ twice). An \emph{unspike} is the reverse operation of a spike, i.e., replace $x$-$y$-$x$ by $x$. The other moves are defined in Section~\ref{sec:definitions}.

The topological non-triviality hypothesis is there to exclude some trivial instances of homotopies in the case $u=v$, for example the empty homotopy starting at $u$ and finishing at $v=u$. In order to do so, we define an algebraic flipping number for each face. Assume without loss of generality (by reversing the embedding if necessary) that the vertices and edges on the outer-face are ordered $u, t(uv), v, s(vu)$ in clockwise order. In particular, if they contain at least two edges, the paths $t(uv)$ and $s(vu)$ go clockwise. Then whenever a face is flipped with a move (either using a face-flip or an edge-slide) between $h_i$ and $h_{i+1}$, this flip is called \emph{positive} (respectively \emph{negative}) if the face lies to the right (respectively to the left) of the edge used to flip it (using the direction of the edge induced by the path $h_i$). The \emph{algebraic flipping number} of each face is defined to be the sum of the positive flips minus the sum of the negative flips over all the moves of a homotopy. A homotopy is \emph{topologically non-trivial} if each face except the outer face has algebraic flipping number one. It is easy to prove that a simple homotopy is always topologically non-trivial. Furthermore, it can be proved that when $u \neq v$, this hypothesis is always satisfied (but since we embedded it in the definition, we will not need this fact).


As in the simple case, the \emph{height} of a homotopy is the length of the longest path $h_i$,
counting as path-length the number of vertices, and $\HH(G)$ is the minimum height of a homotopy over all the possible choices of outer-faces and outer-edges.
\fi

\begin{lemma} For any triangulated planar graph we have $gmh(G) \leq \HH(G)$.
\end{lemma}
\iffull
\begin{proof}
  The proof proceeds similarly to the proof of Lemma~\ref{lem:HHsm_gms}. There are two main differences: we need to explain how to handle the edge-spikes, and the proof that we obtain a \gridminor representation is more involved, since we do not have simplicity to leverage. We first explain how to build the grid and defer the proof that it is a \gridminor representation to the end of the proof.

  The proof starts exactly the same way as that of Lemma~\ref{lem:HHsm_gms}. We build inductively a contact representation following the homotopy, and we use exactly the same notations. Having built a contact representation $\Gamma_i$, we explain how to expand $\Gamma_i$ rightwards to represent $h_{i+1}$. Note that since the \gridminor representation may not be simple, the polygons may not be $x$-monotone, and thus may have several right and left boundaries.

  \begin{itemize}
  \item The cases of edge-slides, boundary edge-slides, boundary moves and face flips are identical to the other proof.
  \item Let us assume that $h_{i+1}$ was obtained from an edge spike (the unspike case is symmetric and is handled identically) along an edge $(x,y)$, i.e., $h_i$ contained $x$ which got replaced by $x$-$y$-$x$ in $h_{i+1}$. So $|h_i| \leq k-2$ which implies that at least one polygon $P(d)$ for some $d \in h_i$ has height at least $3$ on one of its right boundaries, or at least two polygons $P(d_1)$ and $P(d_2)$ (possibly equal) have height at least $2$ on one of their right boundaries. We insert an (upward or downward) staircase that shifts this extra height to $x$ while keeping all polygons of $h_i$ on the right boundary. This may take up to $2k$ units of width but does not affect the height. Now $P(x)$ has height $3$ on the right boundary, and we can start a $1$-pixel high polygon $P(y)$ for the vertex $y$ right inbetween. 
    \end{itemize}

  It remains to prove that this is a valid \gridminor representation. This is not true in general: it is easy to build examples (for example by spiking an edge and immediately unspiking it) where the sets of vertices with a common label in the grid representation are not connected. What we will show is that we can assume that the homotopy we work with has a special property called \emph{monotonicity}, which requires that:

  \begin{enumerate}
  \item each curve $h_i$ does not cross itself, i.e., it can be locally perturbed in the plane to be simple (in the topological sense),
  \item for any $0\leq i<W$, the move connecting $h_i$ to $h_{i+1}$ uses an edge or a face lying to the right of $h_i$, for the orientation induced by the ordering of the vertices in a curve from $s(uv)$ to $t(uv)$.
  \end{enumerate}

The second hypothesis implies that the faces are only flipped positively, and thus each face is flipped exactly once, but it is stronger since it also constrains edges.  For example, while a spiked path is non-simple in the graph-theoretical sense, it can be locally (topologically) perturbed to a simple curve, and is thus allowed in a monotone homotopy. On the other hand, spiking an edge and immediately unspiking it is forbidden, since the unspike goes in the opposite direction of the spike. The following claim will be proved in Lemma~\ref{lem:monotonicity}.

  \begin{claim}
If there exists a discrete homotopy of $G$ of height at most $k$, there exists a monotone discrete homotopy of $G$ of height at most $k$.
  \end{claim}

  Assuming that the homotopy we started with is indeed monotone, we can now prove that our construction yields a valid \gridminor representation. This is intuitive but somewhat technical to prove. Denote by $G^*$ the graph dual to $G$, and view this dual graph as a \emph{cross-metric surface} (see, e.g.,~\cite{ce-tnpcs-10}), i.e., it provides a metric for the curves crossing it \emph{transversely} by counting the number of edges that each edge crosses. As is customary for cross-metric surfaces, we puncture the vertex dual to the outer face, so that the cross-metric surface is a topological disk $D$ whose boundary is made of the two subpaths $s(uv)$ and $t(uv)$. Discrete homotopies are defined on cross-metric surfaces by dualizing the moves described in Section~\ref{sec:definitions}, or see~\cite{cmo-coh-18} for more background on discrete homotopies on cross-metric surfaces.

  Our monotone homotopy induces a (dual) discrete monotone homotopy with respect to $G^*$. Now, the advantage of working with the cross-metric setting is that we can realize the slight perturbations of the curves $h_i$ to make them simple. Furthermore, we can interpolate continuously inside each homotopy move in the natural way to obtain a \emph{continuous homotopy}: a continuous map $h:[0,1] \times [0,1] \rightarrow D$ where $D$ is the disk bounded by the outer face, $h(0,\cdot)$ and $h(1,\cdot)$ are infinitesmal paths in the dual faces $u^*$ and $v^*$, and $h(\cdot,0) \in s(uv)$ and $h(\cdot,1) \in t(uv)$. By the monotonicity assumption, the curves $h(t,\cdot)$ are all simple and locally disjoint for small perturbations of $t$. This implies that the map $h$ is locally injective and therefore a local homeomorphism. The topological non-triviality hypothesis ensures that $h$ has topological degree one, and thus it is a homeomorphism into its image $D$. The preimages of the faces of $G^*$ in $D$ are therefore connected sets in $[0,1] \times [0,1]$. Now the proof follows by observing that the contact representation is exactly a discrete version of these connected sets. Indeed, the labels on the columns of the contact representation mirror the faces crossed by the curves $h(\cdot,t)$, and the construction associated to each move was made to ensure that the connectedness of the labels between two moves in $h$ is faithfully represented by the connectedness of the grid between the two columns.

\end{proof}

To conclude the proof, we establish the needed monotonicity property in the following lemma. The gist of it is that it was proved in~\cite{cmo-coh-18} that optimal homotopies can be assumed to be monotone, but for a slightly different notion of discrete homotopy which did not allow edge slides. In order to use the result of~\cite{cmo-coh-18}, we therefore first need to change the graph $G$ into a different one. Note that since allowing edge-slides change the value of the height by at most $1$, the reader content with this small leeway can apply directly the results of~\cite{cmo-coh-18} to obtain monotonicity.

\begin{lemma}\label{lem:monotonicity}
If there exists a discrete homotopy of $G$ of height at most $k$, there exists a monotone discrete homotopy of $G$ of height at most $k$.
\end{lemma}

\begin{proof}

We denote by $R(G)$ the \emph{radial graph} of $G$, which is obtained by putting a vertex in the middle of each inner face, connecting each pair of vertices lying on adjacent pairs (vertex,face), removing the inner edges and subdividing each outer edge once. The resulting graph $R(G)$ is a bipartite quadrangulation, and discrete homotopies on $R(G)$ are defined with discrete moves as in $G$: now the faces of degree $4$ can be flipped, and that there are no edge-slides since these do not make sense in a quadrangulation (see e.g.,the introduction of~\cite{cmo-coh-18} for an inventory of the discrete homotopy moves in general non-triangulated graphs). The outer face is not allowed to be flipped, and for any subdivided outer edge $e$, we allow a \emph{double boundary edge-slide} to slide in one step over the two corresponding edges in $R(G)$.

We claim that any discrete homotopy of $G$ of height at most $k$ between two vertices $u$ and $v$ induces a discrete homotopy of $R(G)$ of height at most $2k$ between the vertices $u$ and $v$ of $R(G)$. Indeed, any edge $e$ in $G$ can be pushed on either side to yield a $2$-edge path going through either of the two adjacent faces of $e$ (these perturbations are called \emph{vibrations} in~\cite{fomin2006new}). Furthermore, the choice of which way to push does not matter since one can switch between both ways without occuring any increase in the height by doing a face-flip at the face of $R(G)$ corresponding to $e$. Now, simply observe that a spike, face-flip, edge-slide, boundary-move and boundary edge-slide in $G$ corresponds respectively to two consecutive edge spikes, one edge spike, one face-flip, two boundary-moves and one double boundary edge-slide in $R(G)$.

Conversely, any discrete homotopy of $R(G)$ of height at most $2k$ between two vertices $u$ and $v$ that are also vertices in $G$ induces a discrete homotopy of $G$ of height at most $k$ between $u$ and $v$. Indeed, the graph $R(G)$ is bipartite and thus any curve in $R(G)$ can be decomposed into pairs of adjacent edges connecting vertices of $G$. Each of these pairs can be pushed towards $G$ by doing the reverse of the vibrations described above, and the dictionary between moves of $G$ and $R(G)$ can be read in reverse to extract a discrete homotopy of $G$ of height at most $k$ between $u$ and $v$.

Now the proof of Lemma~\ref{lem:monotonicity} simply follows from Theorem~4 of~\cite{cmo-coh-18} showing that optimal homotopies can be assumed to be monotone. More precisely, while Theorem~4 handles homotopies between two cycles forming the boundary of an annulus, the reduction of Proposition~13 in that paper explains how to simply obtain a similar monotonicity property in the setting that we are working with in this paper, where curves are paths with moving endpoints on two boundaries. This reduction naturally relies on the double boundary-edge-slide move, corresponding to a face-flip with the new vertex on the outer face. \end{proof}

We now prove the reverse inequality.

\fi
\begin{lemma}
For any triangulated graph we have $\HH(G) \leq gmh(G)$.
\end{lemma}
\iffull
\begin{proof}
  Here as well, the proof is very similar to that of Lemma~\ref{lem:gms_HHsm}. We will reuse the parts that can be used verbatim, but need to do some slight changes at many places. We start from a \gridminor representation which we interpret as a contact representation, and use the same notations as in Lemma~\ref{lem:gms_HHsm}: for any vertex $z$, $P(z)$ is an orthogonal polygon (but it may not be $x$-monotone).

  We will build the homotopy from this contact representation by intersecting it with the vertical lines, but must first modify it to satisfy additional properties.

Claim~\ref{cl:vertical_junction} carries verbatim and the proof becomes even easier: since there are no $x$-monotonicity constraints, any of the two ways works.
Claim~\ref{cl:aligned_sides} also carries verbatim, so all vertical sides
have distinct $x$-coordinates unless they are on the left or right boundary.  
Also observe that no polygon $P(x)$ can attach at the left boundary in two 
disjoint sides (else $x$ would be a cutvertex), so the leftmost side of a 
polygon $P(x)$ is now defined even though $P(x)$ need not be $x$-monotone.
Claim~\ref{cl:three_outside} used Claim~\ref{cl:two_left}, for which 
$x$-monotonicity was crucial, so we prove a weaker version of this claim in a different fashion.

\begin{claim}
We may assume that the union of the four boundaries consists of either two or three vertices, while preserving the previous assumption.
\end{claim}
\begin{proof}
 First observe that there cannot be more than three different vertices occupying the left, top, right, and bottom boundaries of $\calR$, since $G$ is triangulated. 

 We first assume that a single vertex $u$ occupies the entire left, top right and bottom boundaries of $\calR$. Let $v\neq u$ be a vertex that minimizes the $x$-coordinate of its leftmost side $e_v$.  
Side $e_v$ lies in the interior by $v\neq u$, and cannot contain a junction (not even at its ends), else $v$ would touch a boundary or some vertex $w\neq u,v$ would be farther left than $v$.  Therefore both ends are corners, and convex for $P(v)$ since $e_v$ is leftmost.  Removing everything to the left of the line $\ell$ through $e_v$ then gives a contact-representation of the same graph, since $P(u)$ remains connected (via the top,
right and bottom boundary) and the adjacency $(u,v)$ is realized at the horizontal sides.
Therefore we can assume that at least two vertices occupy the four boundaries.




 If necessary, we reapply the previous claims to ensure the absence of vertical junctions and of two inner vertical sides not on the same boundary having the same $x$-coordinates.
\end{proof}

\begin{claim}
We may assume that exactly one vertex touches the left boundary and exactly one vertex different touches the right boundary, while preserving the previous assumptions.
\end{claim}

\begin{proof}
  Let us assume that there is more than one vertex touching the left boundary. If all the outer vertices touch the left boundary and one of them, say, $w$ occupies both the top left and bottom left corner, then it also occupies the entire top, right and bottom boundaries. Without loss of generality, assume that $u\neq w$ has a minimal $y$-coordinate on the leftmost column. Then we can replace all the pixels below that vertex by $u$. This does not break connectivity of $v$ and also does not remove any adjacencies: indeed, since there are no interior vertical junctions, the part of $P(w)$ on the leftmost column was adjacent to at most one polygon which was not $P(u)$, and it is still adjacent to it after the replacement.

  Thus we can assume that the two vertices adjacent to the top left and bottom left corners are distinct. If there is a third vertex adjacent to the left boundary, choose it. Otherwise, choose arbitrarily one of the vertices on the left boundary. Denoting by $u$ the chosen vertex, we append a new column to the left of $\calR$ and we extend $u$ to the left so that it fills entirely that column. This does not break any boundary-adjacency since the vertices that were on the top left and bottom left corners are still adjacent to the boundary.

  If necessary, repeat on the right boundary, and reapply the previous claims to remove interior junctions and interior vertical sides not both on the same boundary and with the same $x$-coordinates.

\end{proof}

Note that in contrast to Claim~\ref{cl:singletons}, we do not prove that the two vertices can be assumed to be different.

We can now build the homotopy exactly as in the proof of Lemma~\ref{lem:gms_HHsm}. If there are three vertices on the four boundaries of the contact representations, we fix these three vertices to be the outer-face. If there are two vertices, they form the outer-edge. Then the only change is case $(b)$, which is now allowed:

\begin{itemize}
\item[(b)] Both ends of $e$ are corners, and the adjacent horizontal sides go in the same direction. Say that they both go right, the other case being symmetric. Denoting by $P(x)$ the locally convex polygon and by $P(y)$ the other one, this corresponds to a spike: the change from $h_i$ to $h_{i+1}$ is to replace $y$ by $y-x-y$.
\end{itemize}

Finally, in the contact representation, the faces correspond to junctions, and when the induced homotopy is a face-flip or an edge-slide, the faces are swept positively. Since there are no loops nor multiple edges (except perhaps the outer-edge), each junction appears once, and thus each face is swept once. Therefore the homotopy is topologically non-trivial, which concludes the proof.

 \end{proof}
\fi

Since \gridminor height is trivially minor-closed, testing whether a graph has \gridminor height at most $k$ can be decided in time $O(f(k) |G|^3)$ by testing the (unknown!) forbidden minors, which are in finite number by Robertson-Seymour theory. Because minor testing can be expressed
in second-order logic, and the graphs of bounded \gridminor height have bounded pathwidth,
it follows from Courcelle's theorem~\cite{Courcelle} that these minors can be tested in linear time.
 Therefore, the two previous lemmas give us the following corollary.

\begin{corollary}
\label{lem:HH_FPT}
We can decide whether a triangulated planar graph has homotopy height $\HH(G)$ at most $k$ in time $O(f(k) {\mathit poly}(|G|))$ for some computable function $f(k)$. In particular, the problem of computing the homotopy height is FPT when parameterized by the output.
\end{corollary}

\section{Strictness examples}

We have now given all the inequalities needed for
Equations~\ref{equ:result} and \ref{equ:result2}.  In this
section, we argue that many of these inequalities are strict
by exhibiting suitable planar triangulations.

\iffull
\subsection{Pathwidth vs.~\Gridminor height} 

Recall that $\pw(G)\leq \gmh(G)$ since a grid of height $k$ has pathwidth $k$.
We now show that this is strict.
\fi

\begin{lemma}
There exists a planar triangulated graph $G$ with $pw(G)=3$ and
$\gmh(G) \in \Omega(n)$.
\end{lemma}
\begin{proof}
Graph $G$ is the ``nested triangles graph'' from \cite{DLT84,FPP88}, consisting
of $n/3$ triangles that are stacked inside each other and connected 
in such a way that the result is triangulated and has pathwidth 3.
See Fig.~\ref{fig:badGraphs}(a).  For any choice of outer-face there are
at least $n/6$ triangles that remain stacked inside each other.
Therefore $op(G)\geq n/6$ and $\gmh(G)\geq n/3-1$.
\end{proof}

\iffull
\subsection{\Gridminor height, simple \gridminor height and outer-planarity}

Directly from the definition we have $\gmh(G)\leq \sgmh(G)$.
We now show that this can be strict.
\fi

\begin{lemma}
\label{lem:gmh_sgmh}
There exists a planar triangulated graph that has \gridminor height at most 4,
but simple \gridminor-height $\Omega(n)$.
\end{lemma}
\begin{proof}
Consider graph $G$ in Fig.~\ref{fig:badGraphs}(b),
which is taken from \cite{Bie11}.  It is a minor of the graph 
$G'$ in Fig.~\ref{fig:badGraphs}(c), which has a straight-line
drawing of height 4.
Therefore $\sgmh(G')\leq \SLh(G')\leq 4$, which implies 
$\gmh(G)\leq 4$ since $G$ is a minor of $G'$.

\begin{figure}[t]
\hspace*{\fill}
\begin{minipage}{0.3\linewidth}
\centering
\includegraphics[width=0.99\linewidth,page=3]{contractHeight.pdf}

(a)
\end{minipage}
\hspace*{\fill}
\begin{minipage}{0.3\linewidth}
\centering
\includegraphics[width=0.99\linewidth,page=1]{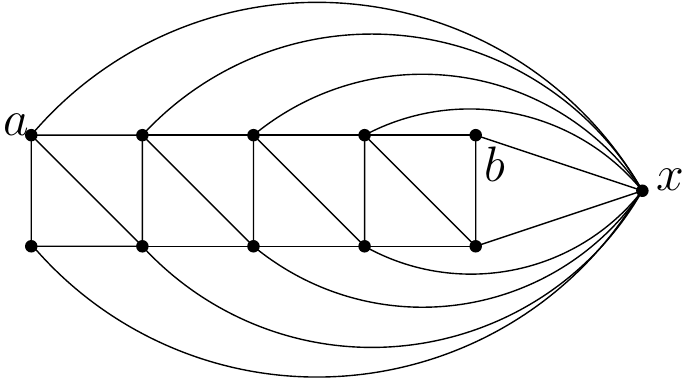}

(b)
\end{minipage}
\hspace*{\fill}
\begin{minipage}{0.3\linewidth}
\centering
\includegraphics[width=0.99\linewidth,page=2]{contractHeight.pdf}

(c)
\end{minipage}
\hspace*{\fill}
\caption{(a) Nested triangles.  (b) The graph from \cite{Bie11}.
(c) A graph of which it is a minor.}
\label{fig:contractHeight}
\label{fig:badGraphs}
\end{figure}

We claim that $\sgmh(G)\in \Omega(n)$, and prove this by 
arguing that $\sHH(G)\in \Omega(n)$; the two parameters are the same.
Crucial to our argument is that for many vertex-pairs
any path connecting them \emph{without} using $x$ has length 
$\Omega(n)$; we will find such a path from the curves in a homotopy.

So consider a simple discrete homotopy of height $k$, and let $f$
be the face it uses as the outer-face (it need not be
the outer-face used in Fig.~\ref{fig:badGraphs}(b)).  Graph $G\setminus x$
is connected, but $d_{G\setminus x}(a,b)=
(n-3)/2$.  Define $d_a$ to be the minimum distance in $G\setminus x$ 
from $a$ to some vertex on face $f$, and similarly define $d_b$.
Since $f$ is a triangle, we can combine two such shortest paths to
obtain a path from $a$ to $b$ in $G\setminus x$ of length at most
$d_a+d_b+1$, therefore (up to renaming) $d_a\geq (n-5)/4$.

In particular, for $n\geq 7$ vertex $a$ is not on $f$. 
Let $h_i$ be a curve of the homotopy that contains $a$ and note
that it begins and ends on $f$.
Split $h_i$ into two paths $\pi_1$ and $\pi_2$ at vertex $a$.
These paths are vertex-disjoint except for $a$ since
the homotopy is simple. 
At most one of these paths 
contains $x$. 
Say $\pi_1$ does not contain $x$ and hence connects $f$
to $a$ without visiting $x$.  Therefore $|\pi_1|\geq d_a \geq (n-5)/4$,
and the height of the homotopy is $\Omega(n)$.
\end{proof}

In particular, Lemma~\ref{lem:gmh_sgmh} provides a different
(and in our opinion more accessible) proof that the graph
in Fig.~\ref{fig:badGraphs}(b) requires $\Omega(n)$ height in
any straight-line drawing \cite{Bie11}.

\iffull
\subsection{\Gridminor height and  graph-drawing height} 

Recall that $\sgmh(G)\leq \VRh(G)=\SLh(G)$ since a visibility
representation can easily be turned into a simple \gridminor representation.  
We now show that this is strict.  We need a definition.  

\fi
A graph $G$ is called a \emph{series-parallel
graph (with terminals $s$ and $t$)} if it either is an edge $(s,t)$, or
if it was obtained via a combination in series or in parallel.  Here,
a \emph{combination in series} 
takes two such graphs $G_i$ with terminals $s_i,t_i$
for $i=1,2$, and identifies $t_1$ with $s_2$.  A \emph{combination in parallel}
also takes two such graphs and identifies $s_1$ with $s_2$ and $t_1$ with
$t_2$.  It is well-known that such graphs are planar.

\begin{lemma}
\label{lem:SP_gm}
Any series-parallel graph has a simple \gridminor
representation of height $O(\log n)$.
\end{lemma}
\begin{proof}
Roughly speaking, we ``bend'' some of the bars in the visibility 
representations of series-parallel graphs from \cite{Bie11} to guarantee
logarithmic height.  Formally we proceed by induction on $m$,
and prove that if $G$ has $m$ edges, then it 
has a simple \gridminor representation of
height $2 \lceil \log m \rceil + 2$ where the
top-right corner is labelled $s$ and the bottom-right corner is labelled $t$.
Furthermore, any column that contains $s$ and/or $t$ also has its 
topmost/bottommost grid point labelled with $s$/$t$. 
In the base case $G$ is an edge $(s,t)$ and we can simply label a $1\times 2$
grid with $s$ and $t$.

Assume first that $G$ was obtained by parallel combinations of $G_1$ and $G_2$.
Consider Fig.~\ref{fig:gmSP}.
After renaming we may assume $m(G_2)\leq m(G_1)$, so $m(G_2)\leq m/2$.
Recursively obtain a \gridminor representation $\Gamma_1$ of 
$G_1$, and pad it with duplicate rows (if needed) so that it 
has height $2\lceil \log m \rceil +2$.
Recursively obtain a \gridminor representation $\Gamma_2$ of 
$G_2$  of height at most $2\lceil \log(m(G_2)\rceil +2 
\leq 2\lceil \log m \rceil$.  Place $\Gamma_2$
to the right of $\Gamma_1$, leaving the top and bottom row unused. 
Label the points above $\Gamma_2$ with $s$ and the points below $\Gamma_2$
with $t$ and verify all conditions.  

Now assume that $G$ was obtained by a series combination of two
graphs $G_1,G_2$ where $G_1$ had terminals $s,x$ and $G_2$ had
terminals $x,t$.  We assume $m(G_2)\leq m(G_1)$, the other case is
symmetric.
Recursively obtain \gridminor representations
$\Gamma_1$ and $\Gamma_2$ of $G_1$ and $G_2$ of
height $2\lceil \log m \rceil+2$ and $2\lceil \log m \rceil$ as before.
Place $\Gamma_2$ to the right of $\Gamma_1$, leaving the top two rows
unused, and leaving one column between the representations unused.
All grid-points in this column, as well as in the row above $\Gamma_2$, are
labelled $x$.  
(In particular, the grid-points labelled $x$ form an ``$S$-shape''
as if we had bent a bar in the middle.)  The second row above $\Gamma_2$
is labelled $s$ so that again the top-right corner has $s$.
One easily verifies that this is a simple \gridminor representation of $G$
with height $2\lceil \log m \rceil +2 \in O(\log n)$.
\end{proof}

\begin{figure}[ht]
\hspace*{\fill}%
\iffull%
\includegraphics[width=0.4\linewidth,page=1]{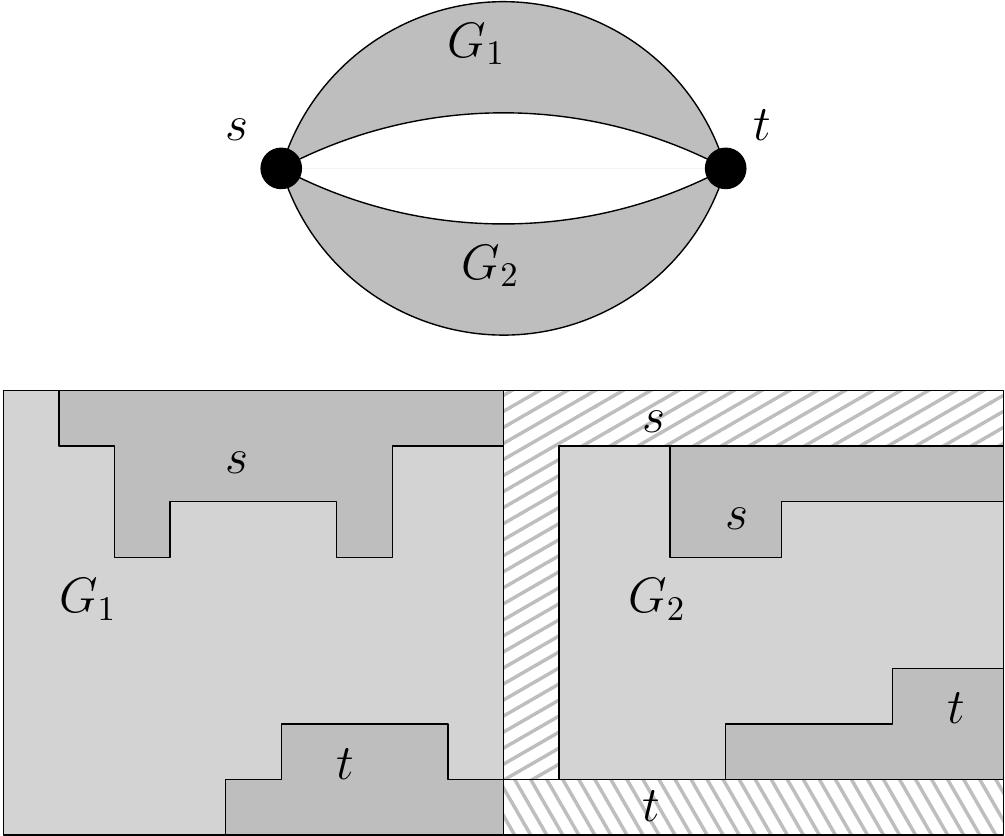}%
\fi%
\hspace*{\fill}%
\iffull%
\includegraphics[width=0.4\linewidth,page=2]{gridMinorSP.pdf}%
\fi%
\hspace*{\fill}%
\caption{\Gridminor representations of series-parallel graphs.}
\label{fig:gmSP}
\end{figure}

\begin{theorem}
\label{thm:gms_h}
There exists a planar triangulated graph $G$ for which
$\sgmh(G)\in O(\log n)$ but $\SLh(G) \in \Omega(2^{\sqrt{\log n}})$.
\end{theorem}
\begin{proof}
We know from Frati \cite{Frati10} that
for any $N$, there exists a series-parallel graph $G_N$ with $n\geq N$
vertices for which any planar straight-line drawing has height
$\Omega(2^{\sqrt{\log n}})$.  
\iffull
Let $\Gamma_N$ be a
simple \gridminor representation of $G'$ of height $O(\log n(G'))$.
Let $\Gamma_N'$ be the simple \gridminor representation obtained
by adding new grid-lines on all four sides of $\Gamma_N$, and by labelling
the top and right with a new vertex $s'$ and the remaining gridpoints
with another new vertex $t'$.  If $\Gamma_N'$ induces a face of degree 4,
then one interior face of the grid has four distinct labels; we can
duplicate the left column and change one label so that an extra edge is
added to the represented graph.  The final result is a grid-representation
of a planar graph $G$ that is triangulated except its outer-face is a duplicate
edge $(s',t')$.  After deleting one copy, we get the desired graph: $G$
has $n+2$ vertices, a simple \gridminor representation $\Gamma$
of height $O(\log n)$,  and contains $G_N$ as a subgraph, so any
straight-line drawing of $G$ has height
$\Omega(2^{\sqrt{\log n)}})$.
\fi
\end{proof}

\begin{lemma}
\label{lem:const-op-log-gmh}
There exists a planar triangulated graph $G$ with $op(G)=2$
but $\gmh(G)\in \Omega(\log n)$.
\end{lemma}
\begin{proof}
Take any tree $T$ that has pathwidth $\Omega(\log n)$, for example
a complete binary tree.  This is an outer-planar graph; add edges
to the graph while maintaining outer-planarity until the graph is
maximal outer-planar, hence 2-connected and all faces except the
outer-face are triangles.  Insert a new vertex in
the outer-face and make it adjacent to all other vertices; the
result (see Fig.~\ref{fig:const-op-log-gmh}) is a triangulated planar graph $G$ with outer-planarity 2
and $\gmh(G)\geq \pw(G)\geq \pw(T) \in \Omega(\log n)$.
\end{proof}

\begin{figure}[t]
\centering\includegraphics[width=0.8\textwidth]{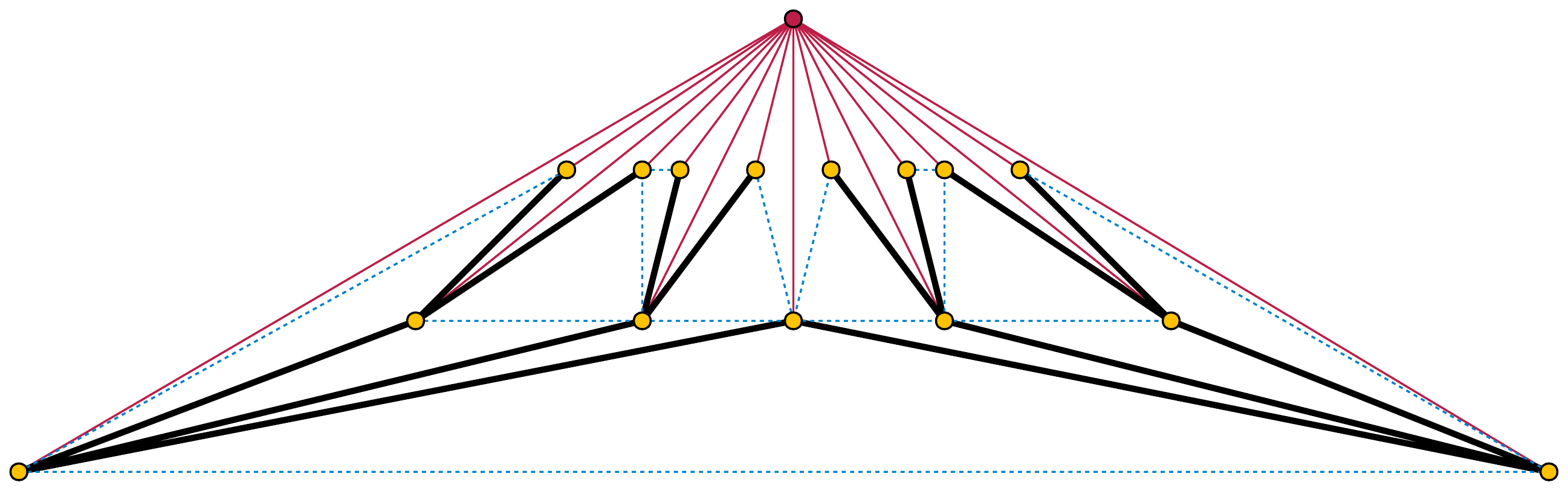}
\caption{A graph with $op(G)=2$ but $\gmh(G)\in \Omega(\log n)$ (Lemma~\ref{lem:const-op-log-gmh}):
a complete binary tree (thick black edges), augmented to become maximal outer-planar (dashed blue edges), with a new vertex added in the outer face (thin red edges).}
\label{fig:const-op-log-gmh}
\end{figure}

\iffull
\section{Algorithms for simple \gridminor height}  
\label{sec:alg-sgmh}
In this section we develop a fixed-parameter tractable algorithm for simple \gridminor height.
Our algorithm uses Courcelle's theorem to recognize the contact-representations of simple \gridminor representations. To do this, we prove that (for fixed values of the height) the following properties of these contact-representations all hold:
\begin{itemize}
\item The contact-representation can be assumed to have only a bounded number of distinct shapes of boundary between pairs of adjacent orthogonal polygons (Lemma~\ref{lem:few-shapes}).
\item The realizability of a single orthogonal polygon, with specified shapes for each of the boundaries with its adjacent polygons, can be expressed in logical terms (Claim~\ref{clm:logical-expression}).
\item A contact-representation with specified boundary shapes exists if and only if each of its polygons is realizable, independently of the others (Lemma~\ref{lem:one-for-all-and-all-for-one}).
\end{itemize}

\subsection{Boundary shapes}

We will assume throughout this section that a contact-representation of height $h$ has vertices with integer $y$-coordinates ranging from $0$ to $h$. However, we allow the $x$-coordinates to be arbitrary real numbers rather than requiring them to be integers. This relaxation has no effect on the existence of contact-representations, but is convenient for us in allowing parts of the representation to be transformed by arbitrary monotone transformations of their $x$-coordinates (keeping the $y$-coordinates unchanged).

We define the \emph{shape} of the boundary $B$ of two polygons in a simple orthogonal contact-representation
to be a polygonal chain $C$ with the following properties:
\begin{itemize}
\item The line segments of $C$ correspond one-to-one with the line segments of $B$, with the same orientations. (We will call these \emph{segments} for short, to distinguish them from the edges of the underlying maximal planar graph.)
\item All $y$-coordinates (heights) of vertices in $C$ equal the coordinates of the corresponding vertices in $B$.
\item The $x$-coordinates of vertices in $C$ are non-negative integers.
\item The $x$-coordinate of the leftmost vertex (or vertices) in $C$ equals zero.
\item When two vertices in $C$ are connected by a horizontal segment in $C$, their $x$-coordinates differ by $\pm 1$.
\end{itemize}

These properties define a unique shape for each polygon-polygon boundary, invariant under $x$-monotone transformations of the boundary.
Our goal in this section is to show that the number of distinct shapes can be bounded by a function of the height~$h$ of the contact-representation.

\begin{lemma}
\label{lem:above-below}
In a simple contact-representation, if two polygons $X$ and $Y$ share a
non-vertical boundary, then one of them is consistently above or below
the other one at all $x$-coordinates shared by both polygons.
\end{lemma}

\begin{proof}
Suppose otherwise, that at  $x$-coordinate $x_1$ polygon $Y$ is below polygon
$X$, and at  $x$-coordinate $x_2$ polygon $Y$ is above $X$. Choose points
$p_i$ in $X$ and $q_i$ in $Y$ (for $i\in\{1,2\}$) with points $p_i$ and $q_i$ having $x$-coordinate $x_i$ respectively.
By the assumption that the contact-representation is simple, 
there exists an $x$-monotone curve $P$ within $X$ from $p_1$ to $p_2$, and another $x$-monotone curve $Q$ within $Y$ from $q_1$ to $Q_2$.
These curves lie within disjoint polygons, so they cannot cross each other.
However at $x_1$, $P$ is below $Q$, and at $x_2$, $Q$ is below $P$.
by the intermediate value theorem there must be an $x$-coordinate between $x_1$ and $x_2$ where they cross. This contradiction shows that inconsistent above-below relations are impossible.
\end{proof}

Define the \emph{extended shape} of a polygon-polygon boundary to be its shape as described above, augmented with a single bit of information that (according to Lemma~\ref{lem:above-below}) describes which of the two adjacent polygons is above and which is below. (For boundaries that have no horizontal segments, we instead use the same bit of information to specify which polygon is to the left of the boundary and which is to the right.)

Given two adjacent polygons $X$ and $Y$, with $X$ above $Y$,
consider the sequence of integer heights of points along the boundary
between $X$ and $Y$, with consecutive duplicates removed, not including
the two points where the boundary begins and ends. Define the sequence
of extrema of boundary $XY$ to be the subsequence of heights that are
either local minima or local maxima of this sequence.

\begin{lemma}
In any sequence of extrema, local minima and local maxima
alternate with each other.
\end{lemma}

\begin{proof}
The extrema are the points where the sequence of differences of
heights between consecutive elements of the height sequence changes
sign. At a local maximum the sequence of differences changes sign from
positive to negative and at a local minimum it changes from negative
to positive. When it changes sign in one direction it cannot change in
the same direction until it has changed back in the other direction.
\end{proof}

We define the \emph{bend complexity} of a contact-representation to be the
sum, over all pairs of adjacent polygons, of the number of bends in the
boundary between the two polygons. Again, we do not count the endpoints of
the boundary as bends.

\begin{lemma}
\label{lem:depocket}
In a simple contact representation of minimum bend complexity
for its height,
it is not possible to have four consecutive extrema in which the outer
two extrema are the maximum and minimum of the four.
\end{lemma}

\begin{figure}[t]
\centering\includegraphics[width=0.8\textwidth]{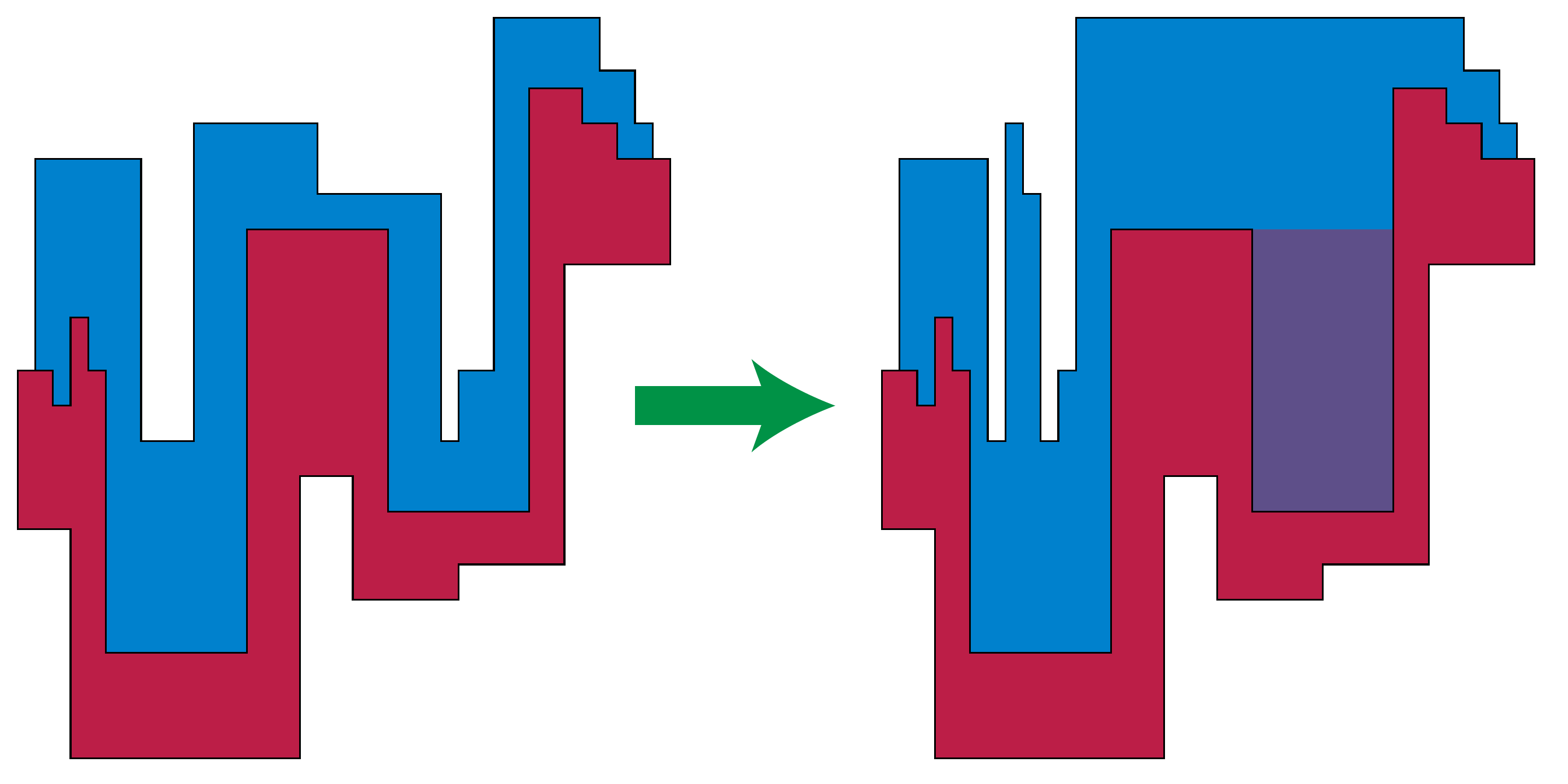}
\caption{Simplifying a boundary that violates Lemma~\ref{lem:depocket}}
\label{fig:depocket}
\end{figure}

\begin{proof}
Proof: Whenever this happens, we could apply a monotonic
transformation to the parts of the contact representation above this
portion of the XY-boundary, leaving the parts below the boundary
untransformed, emptying the shallower of the two pockets in the upper
region and allowing the border to be simplified; see Figure~\ref{fig:depocket}.
\end{proof}

\begin{lemma}
\label{lem:unimodal}
Let $B$ be the boundary between polygons $X$ and $Y$ in a simple contact-representation
of minimum bend complexity for its height.
In the sequence of extrema of heights of the points in $B$, each pair of a global maximum and global minimum of height must be adjacent to each other; hence there can be at most three global extrema, two maxima and one minimum or two minima and one maximum. The local maxima must strictly increase from the start to the first global maximum and strictly decrease after the last global maximum. Symmetrically, the
local minima must strictly decrease from the start to the first global minimum and
decrease after the last global minimum.
\end{lemma}

\begin{figure}[t]
\centering\includegraphics[width=0.6\textwidth]{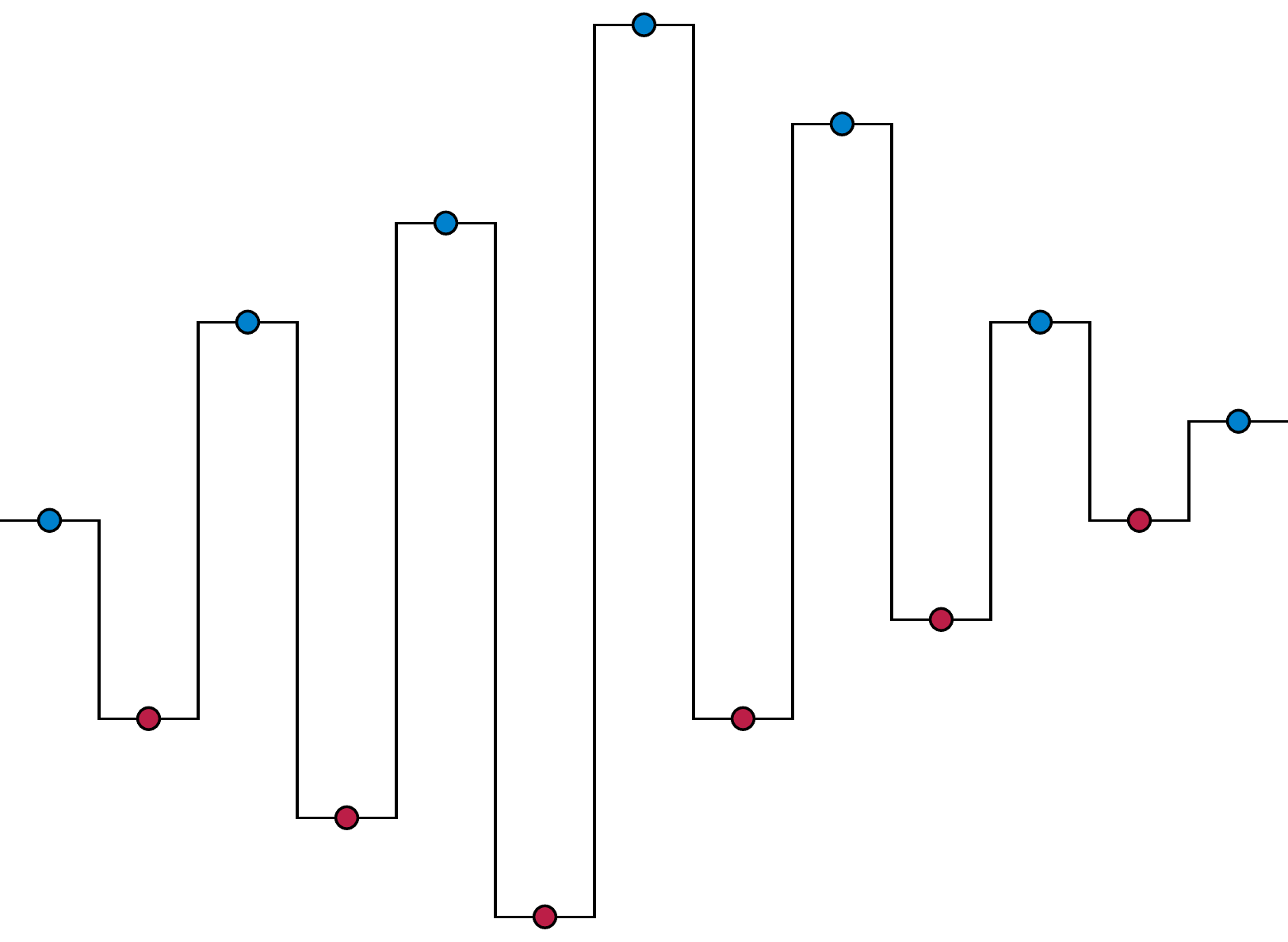}
\caption{Schematic view of the extrema of heights in the boundary of two polygons of a simple contact-representation of minimum bend complexity, according to Lemma~\ref{lem:unimodal}}
\label{fig:unimodal}
\end{figure}

\begin{proof}
Immediately following the global maximum, the next local
maximum must be non-increasing (by global maximality) This forces the
next two local minima to be non-decreasing (by Lemma~\ref{lem:depocket}), which forces
the next two local maxima to be non-increasing, etc. So by induction
the maxima must be non-increasing and the minima must be
non-decreasing from the global maximum to the end of the sequence. The
same argument applies symmetrically between the global maximum and the
start of the sequence.
\end{proof}

For a schematic view of a sequence obeying Lemma~\ref{lem:unimodal}, see Figure~\ref{fig:unimodal}.

\begin{lemma}
\label{lem:few-bends}
In a simple contact representation of minimum bend complexity
for its height~$h$,
the number of bends on the boundary between any two polygons is $O(h^2)$.
\end{lemma}

\begin{proof}
By Lemma~\ref{lem:unimodal} it can have $O(h)$ local extrema of height, and
between any two local extrema there can only be $O(h)$ bends.
\end{proof}

\begin{lemma}
\label{lem:few-shapes}
In a simple contact representation of minimum bend complexity
for its height~$h$, the number of distinct extended shapes of boundaries between
any two polygons is $2^{O(h^2\log h)}$.
\end{lemma}

\begin{proof}
By Lemma~\ref{lem:few-bends} the number of line segments along the boundary is
$O(h^2)$, and each has $O(h)$ possibilities for its height and orientation.
The extra bit of information needed to specify the extended shape from the shape does not affect the result of this calculation.
\end{proof}

\subsection{Single-polygon realizability}

Consider any assignment $A$ of extended boundary shapes to the edges of our given maximal planar graph $G$. We wish to characterize the assignments that are realizable as simple contact-representations; in this section, we do so only for local realizability of a single polygon in the representation. If $A$ is realizable, then each extended shape in $A$, associated to an edge $e$, may be interpreted as describing the shape of the boundary between the two polygons corresponding to the two endpoints of $e$.

The boundary of each polygon in a simple contact-representation has a unique decomposition into four polygonal chains of axis-parallel line segments, which we call the \emph{left}, \emph{right}, \emph{top}, and \emph{bottom} boundary chains:
\begin{itemize}
\item The top boundary chain consists of all horizontal boundary segments that the polygon is below, together with all vertical segments between two of these horizontal segments.
\item The bottom boundary chain similarly consists of all horizontal boundary segments that the polygon is above, together with all vertical segments between two of these horizontal segments.
\item The left boundary chain consists of all vertical segments that the polygon is to the right of and that
are between a horizontal segment of the top boundary and a horizontal segment of the bottom boundary.
\item The right boundary chain consists of all vertical segments that the polygon is to the right of and that
are between a horizontal segment of the top boundary and a horizontal segment of the bottom boundary.
\end{itemize}
Conversely, if a collection of boundary shapes can be realized as a polygon,
and can be partitioned in this way into exactly four chains, then the resulting polygon must be $x$-monotone, suitable for taking part in a simple contact-representation.

Let $v$ be any vertex in graph $G$. We define the \emph{boundary cycle} of $v$ to
be a cycle whose vertices are neighbors of $v$ in $G$,
and whose edges belong to face triangles of $G$ that are incident to $v$.
That is, the boundary cycle connects the neighboring vertices of $v$ in their cyclic order around $v$. It is not necessarily an induced cycle, because $G$ can also include edges between non-consecutive neighbors of $v$ (forming, with $v$, non-facial separating triangles in $G$).
The definition of a boundary cycle is somewhat counterintuitive, because the shapes of the boundary of $v$'s polygon in a contact-representation are not associated with the cycle edges; they are instead associated with the star of edges incident to $v$. However, the boundary cycle conveys important information about the cyclic ordering of $v$'s neighbors that is missing from the star incident to $v$.
When we express the realizability of a polygon in graph logic, we will need this ordering information, in order to
express conditions involving contiguous subsequences of neighbors by representing these subsequences as paths in the boundary cycle.

We say that $v$ has a \emph{monotone boundary} in an extended shape assignment $A$
when the boundary cycle of $v$ can be partitioned into four paths as above:
\begin{itemize}
\item a path that contains all neighboring vertices $u$ of $v$
such that the extended shape associated with edge $uv$ includes horizontal segments that are above the region for $v$,
\item a path that contains all neighboring vertices $u$ of $v$
such that the extended shape associated with edge $uv$ includes horizontal segments that are below the region for $v$,
\item a path that contains all neighboring vertices $u$ of $v$
such that the extended shape associated with edge $uv$ includes vertical segments to the left of the region for $v$ that are between one above-$v$ horizontal segment and one below-$v$ horizontal segment, and
\item a path that contains all neighboring vertices $u$ of $v$
such that the extended shape associated with edge $uv$ includes vertical segments to the right of the region for $v$ that are between one above-$v$ horizontal segment and one below-$v$ horizontal segment.
\end{itemize}
Note that the extended shape of an edge incident to $v$ may include multiple line segments.
Therefore, even though the four boundary chains of a polygon in a contact-representation are interior-disjoint, the corresponding paths in the boundary cycle of a vertex with monotone boundary might share vertices (but not edges) with each other. Some of the boundary paths may be degenerate (consisting of a single vertex) but it is not possible for three to be degenerate and for the fourth to be a cycle containing all vertices rather than a path, as (by Lemma~\ref{lem:above-below}) the top and bottom boundary paths must be vertex-disjoint. Additionally, for $v$ to have a monotone boundary, we require that the four boundary paths appear in an order consistent with the embedding of $G$: for a vertex corresponding to an interior face of the embedding, the clockwise ordering of these paths should be left, above, right, below, and for the vertex corresponding to the outer face this ordering should be reversed.

From an assignment of extended shapes to edges of $G$, and a partition of the boundary cycle of a vertex $v$ into boundary paths, we may construct polygonal chains, the top, bottom, left or right \emph{boundary shapes} of $v$, by concatenating the extended shapes associated with edges $uv$ for vertices $u$ in the boundary path, keeping only the appropriate subsets of the extended shapes associated with the endpoints of the path.

Given any point $p$ in any extended shape, we can recover from the definition of an extended shape the height $\operatorname{height}(p)$ at which $p$ should be realized.
For a point $p$ in the upper or lower boundary chain of a polygon (or in the upper or lower boundary shape of a vertex of $G$ whose shapes are not yet known to be realizable as a polygon)
we define $\operatorname{left}(p)$ to be the largest height among all points strictly to the left of $p$ in the same part of the boundary.
Symmetrically, we define $\operatorname{right}(p)$ to be the largest height among all points strictly to the right of $p$ in the same part of the boundary.

\begin{lemma}
\label{lem:single-face-realizability}
Let $A$ be an assignment of extended shapes to the edges of a maximal planar graph $G$, and let $v$ be a vertex of $G$. Then there exists an $x$-monotone polygon corresponding to $A$, whose boundary realizes the extended shapes assigned by $A$ to the edges in $G$ incident to $v$,
if and only if the following conditions are met:
\begin{itemize}
\item Vertex $v$ has a monotone boundary.
\item For every point $p$ in an extended shape of the upper boundary of $v$,
there must exist a corresponding point $q$ in the extended shape of the lower boundary of $v$,
such that
\[
\operatorname{height}(p) > \operatorname{height}(q),
\]
\[
\operatorname{left}(p) > \operatorname{left}(q),
\]
and
\[
\operatorname{right}(p) > \operatorname{right}(q),
\]
\end{itemize}
\end{lemma}

\begin{proof}
If there exists an $x$-monotone polygon $P$ realizing the boundary shapes of $v$ in shape assignment $A$, then (as above) this polygon's boundary can be partitioned into above, below, left, and right chains, and the corresponding paths in $G$ show that $v$ has a monotone boundary.
For any point $p$ in any extended shape of the upper boundary of $v$, let $p'$ be any point of the segment of $P$ corresponding to the segment containing $p$ in the extended shape.
Let $q'$ be the point of the lower boundary of $P$ directly below $p'$, and let $q$ be any point of the corresponding segment of one of the extended shapes of the lower boundary of $v$.
Then $p$ and $p'$ have the same heights, as do $q$ and $q'$, so $q$ must be below $p$.
At the point of the lower boundary realizing $\operatorname{height}(\operatorname{left}(q))$,
the point on the upper boundary with the same $x$-coordinate must be even higher,
and to the left of $p'$, so the second of the three inequalities above must be valid.
The third inequality must also be valid by symmetric reasoning. Therefore, the existence of $P$ implies that the conditions of the lemma are all met.

Conversely, suppose that assignment $A$ satisfies all the conditions of the lemma for the boundary of $v$. We must show that in this case there exists an $x$-monotone polygon $P$ realizingthe boundary shapes of $v$ in shape assignment $A$.
To do so, we first fix an arbitrary realization for the left, right, and lower boundaries of $P$, by concatenating together the shapes of the boundary segments. By monotonicity, this concatenation cannot be self-crossing. It remains to realize the upper boundary, consistently with this realization of the other parts of the boundary.

To do so, we first perturb the horizontal segments of the lower boundary upwards or downwards by real numbers less than $1/2$, so that no two segments have the same height as each other. Because this perturbation is less than the unit amount by which top and bottom segments of the boundary must clear each other, it has no effect on realizability.
Next, we choose for each horizontal segment $u_i$ of the upper boundary an interior point $p_i$,
assign $p_i$ a point $q_i$ on a horizontal segment of the lower boundary that meets the conditions of the lemma,
and let $\ell_i$ be the segment of the lower boundary containing $q_i$.
By the conditions of the lemma, at least one point $q_i$ exists,
and we choose $q_i$ to be any of the points that match the conditions of the lemma and have the minimum possible height. Because of the perturbation, $\ell_i$ is uniquely defined in this way from $u_i$.

We claim that, for each two consecutive segments $u_i$ and $u_{i+1}$ of the upper boundary,
with $u_i$ to the left of $u_{i+1}$, the corresponding lower boundary segments $\ell_i$ and $\ell_{i+1}$ have the same left-right relation to each other. For, if $u_i$ is lower than $u_{i+1}$, we have that $\operatorname{left}(p_i)\le\operatorname{left}(p_{i+1})$ (because additional points on $u_{i+1}$ itself can contribute to the maximization in the definition of $\operatorname{left}(p_{i+1})$)
but $\operatorname{right}(p_i)=\operatorname{right}(p_{i+1})$ (the additional points on $u_i$ can never have maximum height among all points to the right of $p_i$). Therefore, the left constraint on $p_{i+1}$ is relaxed, allowing $q_{i+1}$ to be farther to the right, while the right constraint is unchanged. A symmetric argument applies to the case when $u_i$ is higher than $u_{i+1}$. In this way, the assignment of lower segments $u_i$ to upper segments $\ell_i$ is $x$-monotone.

To complete the realization, we assign each vertex point $p$ in each segment of the upper boundary an $x$-coordinate. To do so, we choose a vertex $p'$ that should have the same $x$-coordinate as $p$. If $p$ is incident to two horizontal segments of the upper boundary, we let $p'=p$; otherwise we let $p'$ be the lower of the two endpoints of the vertical segment incident to $p$.
Let $u_i$ be a horizontal segment incident to $p'$, and let $\ell_i$ be the corresponding lower segment. The mapping from each $p$ to $\ell_i$, defined in this way, associates each vertex of the upper boundary to a segment of the lower boundary, in an $x$-monotone way. We assign each point $p$ an $x$-coordinate interior to the range of $x$-coordinates spanned by $\ell_i$. It is possible to do this in such a way that upper boundary vertices that should have the same $x$-coordinate (because they are connected by vertical segments) share the same $x$-coordinate,  other pairs of upper boundary vertices have distinct $x$-coordinates, and the $x$-coordinates vary monotonically within the set of all upper boundary vertices assigned to the same~$\ell_i$.

In this way we find a monotonic $x$-coordinate assignment that realizes the upper boundary of the polygon for vertex $v$. By monotonicity, the upper boundary cannot cross itself. And by the condition that each upper boundary vertex corresponds to a lower boundary segment of lower height, the realization of the upper boundary remains at each point above the lower boundary, so the upper and lower boundaries cannot cross each other. Therefore, we have constructed a realization of the whole region for $v$ as an $x$-monotone polygon, as desired.
\end{proof}

\subsection{Logical expression}

To recognize the graphs that have simple contact-representations, we will apply Courcelle's theorem~\cite{Courcelle}, according to which every graph property that can described in monadic second-order logic (more specifically a form of this logic called MSO$_2$) can be recognized for graphs of bounded treewidth in linear time. Here, MSO$_2$ is a form of logic in which the variables may represent vertices, edges, sets of vertices, or sets of edges of a graph. The logic allows both universal quantification ($\forall$) and existential quantification ($\exists$) over these variables.
There are three predicates on pairs of variables: equality of variables of the same type ($=$), set membership ($\in$), and incidence between an edge and a vertex (which we represent by the non-standard binary operator $\multimap$). In addition, all of the usual connectives of Boolean logic are available. We will use variables $v_i$ for vertices, $V_i$ for sets of vertices, $e_i$ for edges, and $E_i$ for sets of edges. If $G$ is a graph and $F$ is a formula of this type, then the notation $G\models F$ (``$G$ models $F$'') means that the formula becomes true when quantification and the predicates are given their usual meanings for the vertices and edges of $G$. Because the equality sign has a meaning as a predicate within this logic, we use $\equiv$ to indicate that two formulas are syntactically equal or to assign a name to a formula. When we use such a name within another formula, it means that the definition of that name should be expanded at that point of the formula, eventually producing (possibly after multiple expansions) a formula that uses only the notation described above.

In this logic, for instance, graph $G$ is connected if and only if it cannot be partitioned into two nonempty vertex sets with no edges between them. That is, we can define a formula
\[
\begin{split}
\operatorname{connected}\equiv
\lnot\exists V_1 \Bigl( &
\exists v_2(v_2\in V_1) \wedge \\
& \exists v_3\bigl(\lnot(v_3\in V_1)\bigr)\wedge \\
& \lnot\exists v_4,v_5,e_6\bigl(
v_4\in V_1 \wedge
\lnot(v_5\in V_1) \wedge
e_6\multimap v_4 \wedge
e_6\multimap v_5)\bigr) \Bigr) \\
\end{split}
\]
such that $G\models\operatorname{connected}$
if and only if $G$ is a connected graph. Using similar logic, we can define a formula
$\operatorname{connected\_subgraph}(E_i)$ that is true when the set of edges $E_i$ defines a connected subgraph of the given graph $G$; we need merely add another clause to the third line of the definition of $\operatorname{connected}$ requiring $e_6$ to belong to $E_i$.
Using this formula, we can define another more complicated formula $\operatorname{path}(E_i)$ that is true when $E_i$ is the edge set of a path: a path is a connected subgraph that has two vertices of degree one and all other vertices of degree two, and the degree conditions are straightforward (if a bit tedious) to formulate in MSO$_2$. Similarly a cycle is a connected subgraph in which all vertices have degree exactly two.

A \emph{peripheral cycle} in a graph $G$ is a simple cycle $C$ with the property that, for every two edges $e_1$ and $e_2$ not belonging to $C$, there exists a path in $G$ containing both $e_1$ and $e_2$ whose degree-two vertices are all disjoint from $C$. In a maximal planar graph (or more generally in a 3-connected planar graph) the faces of the unique planar embedding of the graph are exactly the peripheral cycles~\cite{Tut63}. Because the property of being a peripheral cycle admits a simple logical description in MSO$_2$, we can determine within MSO$_2$ which triangles of our given maximal planar graph are faces. Based on that determination we can also construct a formula that is true of an edge set $E_i$ and vertex $v_j$ when $E_i$ is the boundary cycle of $v_j$, and false otherwise.

In order to formulate the characterization of single-polygon realizability from Lemma~\ref{lem:single-face-realizability} in logical terms, we need some way to express an assignment of extended shapes to the edges of $G$. For fixed $h$ there are by Lemma~\ref{lem:few-shapes} only $O(1)$ distinct extended shapes possible, allowing us to express any such assignment within a logical formula with $O(1)$ edge set variables. Possibly the simplest way of doing so is to use one edge set variable for each distinct extended shape, having as its members the edges to which that shape has been assigned. One can quantify over a shape assignment by applying the same quantifier to each of these edge set variables, and then within the quantification using a conjunction with a subformula that requires each edge to belong to exactly one of these sets. With this logical expression of shape assignments, it is again straightforward (but extremely tedious) to formulate a logical expression $\operatorname{height\_triple}(v,A,P,h_1,h_2,h_3)$
that is true when $v$ is a vertex in $G$,
$A$ is an extended shape assignment for $G$,
$P$ is either the upper or lower boundary of $v$ with respect to $A$,
and there exists a point $p$ within the shape of one of the edges of $P$ for which
$\operatorname{left}(p)=h_1$, $\operatorname{height}(p)=h_2$, and $\operatorname{right}(p)=h_3$. Here, the numerical heights are not themselves logical variables; rather, there is one such formula for each of the $O(h^3)$ different possible combinations of heights.

Using this subformula, it is straightforward to express the conditions of Lemma~\ref{lem:single-face-realizability} logically, by asking for the existence of a cycle and four paths that form the boundary cycle for $v$ and the decomposition of this cycle into boundary paths,
and requiring that for these paths and for each realizable triple of heights on the upper boundary there exists a compatible triple of heights on the lower boundary. That is, for each expression $\operatorname{height\_triple}(v,A,P,h_1,h_2,h_3)$ we write a logical implication from the expression applied to the upper boundary to a disjunction of expressions for compatible triples applied to the lower boundary, and we take the conjunction of all such implications.

We summarize the discussion of this section by the following:

\begin{claim}
\label{clm:logical-expression}
For any fixed height $h$ there exists a logical formula
$\operatorname{realizable}(v,A)$
that is true whenever $A$ is an extended shape assignment (described as a system of logical variables) for which there exists an $x$-monotone polygon whose boundary realizes the extended shapes assigned by $A$ to the edges in $G$ incident to $v$.
\end{claim}

\subsection{Global realizability}

As we now prove, local realizability of each polygon in a contact-representation implies global realizability of the entire representation. This local-global principle is analogous to similar local-global principles for realizability of upward planar graph drawings~\cite{BDLM94}, rectilinear planar graph drawings~\cite{Tam87}, level planar graph drawings, and flat-foldable graph drawings~\cite{ADD+18}.

\begin{lemma}
\label{lem:one-for-all-and-all-for-one}
Let $G$ be a maximal planar graph, and suppose that for an extended shape assignment to $G$ of height $h$
each polygon is individually realizable (per Lemma~\ref{lem:single-face-realizability}). Then there exists a simple contact-representation of height $h$ for $G$.
\end{lemma}

\begin{proof}
We may order the vertices of $G$ (and with them, the polygons of the realization) by the heights of their lowest boundary segments in the shape assignment, breaking ties arbitrarily. We will construct the realization (again, allowing real numbers as $x$-coordinates) for each polygon in order from lowest to highest in this order.
By Lemma~\ref{lem:above-below}, the bottom boundary of each polygon $P$ will have already been determined from the realizations of its lower neighboring polygons, at the time we realize polygon $P$.
Given an arbitrary realization of polygon $P$ (assumed to exist by the preconditions of this lemma), we may find an $x$-monotone transformation from that realization to the fixed set of $x$-coordinates of its vertices in the realization of the polygons below $F$. Applying the same transformation to the upper boundary of $P$ produces a realization of $P$ compatible with the lower polygons. This transformed copy of $P$ can be added to the realization of all polygons up to $F$ in the height order. By induction on position within the height order, all interior polygons are simultaneously realizable. Once they are all realized, the exterior polygon is automatically realized as~well.
\end{proof}

\subsection{Recognition algorithm}

\begin{theorem}
There exists an algorithm for recognizing $n$-vertex graphs of simple \gridminor height $h$, and constructing a simple contact-representation for them, in time $O(f(h) n)$ for a computable function $f$.
\end{theorem}

\begin{proof}
We construct the formula $\operatorname{realizable}(v,A)$ for $h$, as described in Claim~\ref{clm:logical-expression},
and from it construct the formula
\[
\operatorname{representable}\equiv\exists A\Bigl(\forall v\bigl(\operatorname{realizable}(v,A)\bigr)\Bigr)
\]
(as before, expanding the mapping $A$ from edges to extended shapes into a large number of edge set variables and expanding the existential quantifier on $A$ into many existential quantifiers for each of these variables).
By Lemma~\ref{lem:one-for-all-and-all-for-one}, $G\models\operatorname{representable}$ if and only if $G$ has a simple \gridminor contact-representation of height~$h$.
Next, by standard methods we check that $G$ has treewidth at most $h$ (true of all graphs of simple \gridminor height $h$) and if so construct a tree-decomposition for $G$. If not, we reject $G$ as not having simple \gridminor height $h$. Finally, we use Courcelle's theorem to check from the formula for $\operatorname{representable}$ and from the tree-decomposition of $G$ whether $G\models\operatorname{representable}$. Standard variations of Courcelle's theorem for formulas whose outer quantifiers are existential, such as the formula for $G\models\operatorname{representable}$, can find values for the quantified variables for which the rest of the formula is true, and we use these to find the extended shapes of the boundaries in a realizable simple contact-representation of $G$. From these, it is straightforward to construct the contact-representation itself, along the lines of the proofs of Lemmas \ref{lem:single-face-realizability} and~\ref{lem:one-for-all-and-all-for-one}.
\end{proof}

Because the logical formula used to express contact-representation realizability has size exponential in $h^2\log h$,
and Courcelle's theorem has non-elementary dependence on formula length,
our simple \gridminor height algorithm is not practical. It remains open to find a practical fixed-parameter tractable algorithm for the same problem.

\fi

\section{Outlook }  

In this paper, we studied two parameters of planar triangulated
graphs, the homotopy height (well-known in computational geometry
but not previously used for graph drawing) and the \gridminor height
(related to contact-representations, but not explicitly expressed as
a graph parameter before).  We argue that these two seemingly unrelated
parameters are actually equal, and that they, as well as their variations
that require simplicity in some sense, can serve as lower bounds for
the height of straight-line drawings of planar graphs.  Their equality
also implies that testing whether homotopy height is at most $k$ 
is fixed-parameter tractable in $k$.   We leave many open problems:

\begin{itemize}
\item What is the complexity of computing these various graph parameters?
	In particular, while it is strongly believed that computing the
	minimum height of a planar drawing is NP-hard, we are not aware
	of any proof of this. 
        Similarly, it is not known whether computing the homotopy height, or equivalently the \gridminor height, is NP-hard or polynomial. The same goes for the simple variants.
	On the other hand, computing the pathwidth is NP-hard even for planar graphs \cite{Gustedt},
	while computing the outerplanarity is polynomial \cite{BienstockMonma}.
      \item The trivial minor-closedness of \gridminor height proves the existence of an FPT algorithm to compute it when parameterized by the output. However, this algorithm relies on finding the forbidden minors, which are unknown. Finding an explicit algorithm for this problem\iffull, for instance along the lines of Section~\ref{sec:alg-sgmh},\fi{} is still open. \iffull And because Section~\ref{sec:alg-sgmh} uses Courcelle's theorem, its algorithm for simple \gridminor height is not very efficient; finding a more efficient algorithm again remains open.\fi

\item We focused on straight-line drawings, but \emph{poly-line drawings}
of $G$ (i.e., straight-line drawings of some subdivision $G'$ of $G$)
are also of interest.  Letting $\PLh(G)$ be the smallest height of such
drawings, one sees that $\gmh(G)\leq \sgmh(G') \leq \SLh(G')\leq \PLh(G)$,
but is it true that $\sgmh(G)\leq \PLh(G)$?



\end{itemize}

\iffull
\bibliographystyle{plainurl}
\fi
\bibliography{biblio}

\end{document}